\pdfoutput=1
\newif\ifARXIV \ARXIVtrue
\documentclass[acmsmall,screen]{acmart}

\AtBeginDocument{%
  }

\setcopyright{rightsretained}
\acmDOI{10.1145/3656419}
\acmYear{2024}
\copyrightyear{2024}
\acmSubmissionID{pldi24main-p276-p}
\acmJournal{PACMPL}
\acmVolume{8}
\acmNumber{PLDI}
\acmArticle{189}
\acmMonth{6}
\received{2023-11-16}
\received[accepted]{2024-03-31}

\usepackage[utf8]{inputenc}
\usepackage[T1]{fontenc}
\usepackage{microtype}
\usepackage{bm}
\usepackage[capitalize,nameinlink]{cleveref}
\usepackage{mathtools}

\usepackage{multirow}
\usepackage{makecell}
\usepackage{thm-restate}
\usepackage{algpseudocode}
\usepackage{algorithm}
\usepackage{subcaption}
\usepackage{nicefrac}

\AtBeginEnvironment{example}{%
  \pushQED{\qed}%
}
\AtEndEnvironment{example}{\popQED\endexample}

\usepackage{enumitem}
\setlist[itemize]{leftmargin=5mm}

\newcommand{\ket}[1]{\lvert #1\rangle}
\newcommand{\bra}[1]{\langle #1\rvert}
\newcommand{\ketbra}[3][]{\ket{#2}_{#1}\bra{#3}}
\newcommand{\braket}[2]{\left\langle #1\middle\vert #2\right\rangle}
\newcommand{\abs}[1]{\left\lvert #1 \right\rvert}
\newcommand{\norm}[1]{\left\lVert #1 \right\rVert}
\DeclareMathOperator{\tr}{tr}
\DeclareMathOperator{\Span}{span}

\newcommand{\CC}{\mathbb{C}}
\newcommand{\ZZ}{\mathbb{Z}}
\newcommand{\FF}{\mathbb{F}}

\usepackage{dsfont}
\newcommand{\Id}{\mathds{1}}

\newcommand{\cD}{\mathcal{D}}
\newcommand{\cH}{\mathcal{H}}

\newcommand{\cS}{{S}}

\newcommand{\cnot}{\mathit{CNOT}}

\usepackage{semantic}
\reservestyle{\command}{\texttt}
\command{measure, if, else, X, Y, Z, CNOT, S, H, true, false}

\newcommand{\assign}[2]{#2 \coloneqq{} #1}
\newcommand{\qut}[2]{#1\ #2}
\newcommand{\qmeasure}[2]{\assign{\<measure>\ #1}{#2}}
\newcommand{\qif}[2]{\<if>\ #1\ \{ #2 \}}
\newcommand{\qqif}[3]{\<if>\ #1\ \{ #2 \}\ \<else>\ \{ #3 \}}
\newcommand{\qprogs}{\mathbf{qProgs}}
\newcommand{\cfg}[3]{\langle #1, #2, #3 \rangle}
\newcommand{\scfg}[5]{\langle #1, #2, #3, #4, #5 \rangle}
\newcommand{\terminate}{\downarrow}
\newcommand{\overto}[2][]{\xlongrightarrow{#2}\!\!{}^{#1}\,}

\newcommand{\ssigma}{\widetilde{\sigma}}
\newcommand{\srho}{\widetilde{\rho}}

\usepackage{tikz}
\usetikzlibrary{positioning, shapes.geometric}
\definecolor{qubitcolor}{RGB}{239,173,62}
\definecolor{zcheckcolor}{RGB}{42, 88, 178}
\definecolor{xcheckcolor}{RGB}{214, 82, 60}
\newcommand{\mathtikz}[2][c]{%
\begin{array}[#1]{c}
\begin{tikzpicture}
    #2
\end{tikzpicture}
\end{array}}

\usepackage{pgfplots}
\pgfplotsset{
    compat=1.5,
    grid=major,
}
\usetikzlibrary{plotmarks}
\usetikzlibrary{quantikz}

\def\COMPILETIKZ{0}

\usetikzlibrary{external}

\ifthenelse{\COMPILETIKZ=1}{
\tikzexternalize[prefix=Tikz/]%
\tikzexternaldisable
}{}

\newcommand{\tikzinput}[2]{
\ifthenelse{\COMPILETIKZ=1}
{%
{
\tikzexternalenable\tikzsetnextfilename{#1}\input{#2}%
}
}%
{\includegraphics{Tikz/#1.pdf}}%
}

\newcommand{\QuantumSE}{\textsf{QuantumSE.jl}}

\usepackage{titlesec}
\newcommand{\periodafter}[1]{#1.}
\titleformat{\subsubsection}[runin]{\normalfont\normalsize\bfseries\boldmath}{\thesubsubsection}{.5em}{\periodafter}
\titlespacing{\subsubsection}{0pt}{*1}{*1}

\begin{document}

\title[Symbolic Execution for Quantum Error Correction Programs]{Symbolic Execution for Quantum Error Correction Programs \\ (Extended Version)}

\author{Wang Fang}
\orcid{0000-0001-7628-1185}
\affiliation{%
  \institution{Institute of Software at Chinese Academy of Sciences}
  \city{Beijing}
  \country{China}
}
\affiliation{%
  \institution{University of Chinese Academy of Sciences}
  \city{Beijing}
  \country{China}
}
\email{fangw@ios.ac.cn}

\author{Mingsheng Ying}
\orcid{0000-0003-4847-702X}
\affiliation{%
  \institution{Institute of Software at Chinese Academy of Sciences}
  \city{Beijing}
  \country{China}
}
\affiliation{%
  \institution{Tsinghua University}
  \city{Beijing}
  \country{China}
}
\email{yingms@ios.ac.cn}


\begin{abstract}
    We define QSE, a symbolic execution framework for quantum programs by integrating symbolic variables into quantum states and the outcomes of quantum measurements.
The soundness of QSE is established through a theorem that ensures the correctness of symbolic execution within operational semantics.
We further introduce symbolic stabilizer states, which symbolize the phases of stabilizer generators, for the efficient analysis of quantum error correction (QEC) programs.
Within the QSE framework, we can use symbolic expressions to characterize the possible discrete Pauli errors in QEC, providing a significant improvement over existing methods that rely on sampling with simulators.
We implement QSE  with the support of symbolic stabilizer states in a prototype tool named \QuantumSE{}.
Our experiments on representative QEC codes, including quantum repetition codes, Kitaev's toric codes, and quantum Tanner codes, demonstrate the efficiency of \QuantumSE{} for debugging QEC programs with over 1000 qubits.
In addition, by substituting concrete values in symbolic expressions of measurement results, \QuantumSE{} is also equipped with a sampling feature for stabilizer circuits.
Despite a longer initialization time than the state-of-the-art stabilizer simulator, Google's Stim, QuantumSE.jl offers a quicker sampling rate in the experiments.

\end{abstract}



\maketitle

\section{Introduction}
\label{sec:intro}


Nowadays, quantum computing hardware is developing rapidly~\cite{arute2019quantum, wu2021strong,madsen2022quantum,bluvstein2024logical}, and the number of its physical qubits is also gradually increasing. For instance,  IBM has introduced the IBM Condor, an advanced quantum processor featuring 1,121 superconducting qubits~\cite{ibm2023roadmap}.
However, these quantum hardware suffer from errors caused by quantum noise and inaccurate quantum gate implementations,
so deploying quantum error correction (QEC) programs on quantum hardware is the key to large-scale quantum computing in the future~\cite{shor1996fault,aharonov1997fault,gottesman2013fault}.
There are already several experimental studies exploring QEC in quantum hardware~\cite{Acharya2022SuppressingQE,ai2021exponential,zhao2022realization,abobeih2022fault,bluvstein2024logical},
among which the recent result of 
\citet{bluvstein2024logical} has successfully encoded 48 logical qubits on neutral atom arrays with up to 280 physical qubits.
On the other hand, as diverse and more complicated QEC protocols are introduced, we will need to conduct prior analysis and verification of them to guarantee their correctness before deployment.
This leads us to the following basic question: 
{
\addtolength\leftmargini{-0.2in}
\begin{quote}
    \emph{How can we check whether a QEC program (e.g., the one depicted in \cref{fig:running_example_decoder}), as a hybrid quantum-classical program, has any bugs, especially the part involving quantum variables; and more importantly, if the QEC program has bugs, how can we automatically find them?}
\end{quote}
}

\begin{figure}[ht]
    \centering
    \scalebox{.9}{\tikzinput{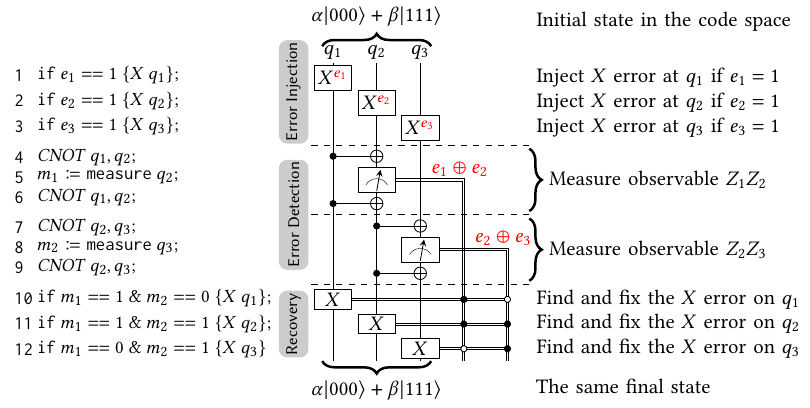}{figures/diagram.tex}}

    \vspace*{-2mm}
    \caption{A QEC program (Lines \texttt{4}-\texttt{12}) for the three-qubit bit-flip code with code distance $3$. Ensuring the program is bug-free requires that for any initial state $\alpha\ket{000}+\beta\ket{111}$ within the code space, the final state remains $\alpha\ket{000}+\beta\ket{111}$ after error injection (Lines \texttt{1}-\texttt{3}) and execution of the QEC program (Lines \texttt{4}-\texttt{12}). Error injection injects errors that the QEC code can tolerate into the quantum state, and then the QEC program corrects these errors. Here, the possible errors correspond to variables $e_1,e_2,e_3 \in\{0,1\}$ that satisfy $e_1+e_2+e_3 \leq 1$ (up to one $X$ error occurs). We also refer to the errors $X^{e_1}, X^{e_2}, X^{e_3}$ introduced during error injection as "adversarial" because of their potential to disrupt the integrity of information in quantum states.}\label{fig:running_example_decoder}
\end{figure}


Fortunately, current QEC programs only involve stabilizer circuits and thus can be efficiently simulated on classical computers~\cite{gottesman1997stabilizer}. Indeed, several specific stabilizer circuit simulators have been  developed~\cite{aaronson2004improved,anders2006fast,gidney2021stim,krastanov} for analyzing QEC circuits.
For instance, Google's recent QEC experiment~\cite{Acharya2022SuppressingQE} utilized a fast simulator, Stim~\cite{gidney2021stim}, to sample a prior distribution for detectors of the experimental QEC circuit.
For the example program in \cref{fig:running_example_decoder}, the current practice involves enumerating/sampling the potential values of $e_i\in \{0,1\}$ such that $e_1+e_2+e_3\leq 1$, then executing the program with them by using a fast simulator. 
However, for a general QEC code with $n$ qubits and a code distance of $d$, the potential errors during error injection correspond to variables $e_1,e_2,\ldots,e_n\in\{0,1\}$ with the condition $e_1+e_2+\cdots e_n\leq \lfloor\frac{d}{2}\rfloor$.
The count of these errors exceeds $\binom{n}{\lfloor\frac{d}{2}\rfloor}$, which goes beyond polynomial-time complexity.
Despite the simulator's rapid performance\footnote{e.g., the simulator Stim~\cite{gidney2021stim} can simulate a distance 100 surface code circuit (20 thousand qubits, 8 million gates, 1 million measurements) in 15 seconds.},
such number of errors would make this approach unscalable at a large scale.
Additionally, as noted by \citet{gidney2021stim}, \emph{generating samples of QEC circuits is the bottleneck in analyzing QEC programs}.

In classical computing, symbolic execution (SE) was proposed~\cite{king1976symbolic} to handle such a large number of samples or test cases.
It is an automated method for finding bugs in programs and has been successfully applied in many different  fields, including software testing~\cite{cadar2013symbolic,godefroid2005dart,cadar2008klee,poeplau2020symbolic} and security and privacy~\cite{farina2019relational,aizatulin2011extracting}, and extended to non-deterministic programs~\cite{luckow2014exact,siegel2008combining,yu2020symbolic} and 
probabilistic programs~\cite{susag2022symbolic,gehr2016psi}.
Recently, SE has also been introduced into quantum computing in several pioneering papers \cite{carette2023symbolic,bauer2023symqv,10.1145/3519939.3523431}.
But all of them only deal with quantum circuits \textit{without control flows}.
As is well-known, one of the most essential advantages of SE is that it can systematically explore many possible execution paths (determined by control flows) at the same time.
Therefore, the power of SE has not been fully exploited in quantum programming.

In this paper, we significantly extend SE to a quantum programming language with \textit{(classical) control flows}, which is particularly suited to writing QEC programs.
\emph{Our SE framework introduces symbols to represent the possible (probabilistic) outcomes of quantum measurements in a uniform manner}, rather than splitting the measurement outcomes into cases as in previous work~\cite{bauer2023symqv}.
This approach circumvents the potential exponential growth of execution paths that comes with quantum measurements, which 
are a fundamental part of QEC programs.
To make our SE framework efficient in handling QEC programs, we further introduce \emph{symbolic stabilizer states} by a partial symbolization of stabilizer states, the underlying quantum states in stabilizer circuits.
Thanks to the fast simulation algorithm of stabilizer circuits~\cite{aaronson2004improved}, our symbolic stabilizer states can be efficiently manipulated during SE.

More importantly, symbolic stabilizer states allow us to handle conditionally applied Pauli gates (see Lines \texttt{1}-\texttt{3} and \texttt{10}-\texttt{12} of \cref{fig:running_example_decoder}) without forking into multiple execution paths: the conditional guards, e.g., $e_1==1$ and $e_2==1$, are absorbed into the symbolic stabilizer state, yielding only a single execution path and transforming subsequent measurement outcomes into symbolic expressions, e.g., $e_1\oplus e_2$ for variable $m_1$ in Line \texttt{5} of \cref{fig:running_example_decoder}.
Consequently, we can use symbolic expressions, e.g., $e_1+e_2+e_3\leq 1$, 
to characterize the possible errors in error injection and use SMT solver to deduce the correctness, obviating the need for enumerating concrete values that satisfy the condition.

\subsubsection*{Contributions and outline}
After reviewing some background knowledge (\S\ref{sec:background}) and demonstrating our approach on a running example (\S\ref{sec:example}), our major contributions are presented as follows:
\begin{itemize}
    \item We develop \textit{a framework for symbolic execution of quantum programs} (QSE) with classical control flows (\S\ref{sec:qse}) and prove a \textit{soundness theorem} for our QSE framework (\S\ref{sec:soundness}).
    \item We introduce symbolic stabilizer states to enable our QSE framework to efficiently analyze QEC programs and prove an \textit{adequacy theorem} for symbolic stabilizer states in the analysis and verification of QEC programs (\S\ref{sec:symbolic_stabilizer}).
    \item We introduce \emph{symbolic Pauli gates} with a newly designed SE rule that can handle the conditional application of Pauli gates without forking like classical SE and demonstrate how it can be used to characterize the possible adversarial errors for QEC programs (\S\ref{sec:conditional}).
    \item We implement our QSE framework together with the special support of symbolic stabilizer states in a \textit{prototype tool} named \QuantumSE{} and demonstrate its efficiency and the ability to outperform existing tools in experimental evaluation (\S\ref{sec:evaluation}).
\end{itemize}
The paper is concluded by discussions about related work (\S\ref{sec:related_work}) and issues for further research  (\S\ref{sec:conclusion}).

\section{Background}
\label{sec:background}
In this section, we provide a minimal background of quantum computation and QEC.
The reader can consult the book~\cite[Chapter 2,4,10]{nielsen2010quantum} for more details.

\subsection{Quantum Preliminary}
\label{sec:preliminary}
We assume basic knowledge of linear algebra, including the concepts of vector space and tensor product.
For a $d$-dimensional complex vector space $\CC^d$, we use the Dirac notation $\ket{\psi}$ to denote a column vector in it.
The conjugate transpose of $\ket{\psi}$ is then a row vector denoted by $\bra{\psi}$.
The \emph{inner product} of $\ket{\psi}$ and $\ket{\phi}$ is a complex number denoted by $\braket{\phi}{\psi}$.
The \emph{outer product} of $\ket{\psi}$ and $\ket{\phi}$ is a $d\times d$ matrix denoted by $\ketbra{\psi}{\phi} \in \CC^{d\times d}$.
The norm of a vector $\ket{\psi}$ is defined as $\norm{\ket{\psi}} = \sqrt{\braket{\psi}{\psi}}$.
In addition, the tensor product of $\ket{\psi_1}$ and $\ket{\psi_2}$ is denoted by $\ket{\psi_1}\otimes \ket{\psi_2}$, which is sometimes written as $\ket{\psi_1}\ket{\psi_2}$ or even $\ket{\psi_1\psi_2}$ for short.

\subsubsection*{Quantum states}
The state space of a quantum bit (qubit) is the  $2$-dimensional vector space $\CC^2$ with $\ket{0} = \begin{psmallmatrix}1\\0\end{psmallmatrix}, \ket{1} = \begin{psmallmatrix}0\\1\end{psmallmatrix}$ being the \emph{computational basis}.
In general, the state space of $n$-qubit system is the tensor product of $n$ copies of the state space $\CC^2$ of a single qubit, which is $(\CC^2)^{\otimes n} \cong \CC^{2^n}$, with $\{\ket{x}|x\in\{0,1\}^n\}$ being the computational basis. A \emph{pure} quantum state is represented by a unit vector $\ket{\psi}$. Thus,  
an $n$-qubit pure state $\ket{\psi} \in (\CC^{2})^{\otimes n}$ can be expressed as $\sum_{x\in\{0,1\}^n} \alpha_x\ket{x}$, where $\alpha_x \in \CC$ and $\sum_{x\in\{0,1\}^n} \abs{\alpha_x}^2 = 1$.
When the state of an $n$-qubit system is not completely known, one may think of it as a \emph{mixed} state (an ensemble of pure states) $\{(p_j,\ket{\psi_j})\}$ meaning that it is in state $\ket{\psi_j}$ with probability $p_j$. Such a mixed state can also be represented by 
a \emph{density operator}  $\rho = \sum_{j}p_j\ketbra{\psi_j}{\psi_j}$, which is a $2^n\times 2^n$ positive semidefinite complex matrix.
In particular, a pure state $\ket{\psi}$ can be represented by the density operator $\ketbra{\psi}{\psi}$; for simplicity, we write $\psi = \ketbra{\psi}{\psi}$.

\subsubsection*{Unitary transformations}
A $2^n\times 2^n$ complex matrix $U$ is called unitary if $U^{\dagger}U = \Id_{2^n}$, where $U^\dag$ stands for the conjugate transpose of  $U$, and $\Id_{2^n}$ denotes the $2^n\times 2^n$ identity matrix. It models a transformation or an  evolution of an $n$-qubit system from a pure state $\ket{\psi} \in \CC^{2^n}$ to $U\ket{\psi}$.
For mixed states, it transforms a density matrix $\rho \in \CC^{2^n\times 2^n}$ to $U\rho U^{\dagger}$.
Such a unitary transformation $U$ is often called an $n$-qubit gate.  
Common quantum gates include the single-qubit gates $H$ (Hadamard gate), $S$ (Phase  gate), and $I =\Id_2, X,Y,Z$ (Pauli gates) as well as the $2$-qubit gate $\cnot$ (controlled-NOT gate): 
\begin{gather*}
    H=\frac{1}{\sqrt{2}}\begin{pmatrix*}[r]
        1 & 1 \\
        1 & -1
    \end{pmatrix*},\ 
    S = \begin{pmatrix*}[r]
        1 & 0 \\
        0 & i
    \end{pmatrix*},\ 
    X = \begin{pmatrix*}[r]
        0 & 1 \\
        1 & 0
    \end{pmatrix*},\ 
    Y = \begin{pmatrix*}[r]
        0 & -i \\
        i & 0
    \end{pmatrix*},\ 
    Z = \begin{pmatrix*}[r]
        1 & 0 \\
        0 & -1
    \end{pmatrix*},\ 
    \cnot = \begin{psmallmatrix*}[r]
        1 & 0 & 0 & 0 \\
        0 & 1 & 0 & 0 \\
        0 & 0 & 0 & 1 \\
        0 & 0 & 1 & 0 
    \end{psmallmatrix*}
\end{gather*}
The $H$ gate can produce another important basis, called the $\pm$ basis, $\{\ket{+} = \frac{1}{\sqrt{2}}(\ket{0}+\ket{1}), \ket{-} = \frac{1}{\sqrt{2}}(\ket{0}-\ket{1})\}$ of $\CC^2$ from the computational basis as $H\ket{0} = \ket{+}$ and $H\ket{1} = \ket{-}$.
The Phase gate $S$ leaves $\ket{0}$ unchanged and adds a phase of $i$ to $\ket{1}$, i.e., $S\ket{0}=\ket{0}$, $S\ket{1} = i\ket{1}$; similarly, the Pauli $Z$ gate only adds phase $-1$ to $\ket{1}$, i.e., $Z\ket{0} = \ket{0}, Z\ket{1}=-\ket{1}$.
The Pauli $X$ gate acts like a ``NOT'' gate, exchanging $\ket{0},\ket{1}$, i.e., $X\ket{0} = \ket{1}$ and $X\ket{1} = \ket{0}$;
similarly the Pauli $Y$ gate switches $\ket{0}, \ket{1}$ and introduces additional phases of $i$ and $-i$, i.e., $Y\ket{0} = i\ket{1}, Y\ket{1}=-i\ket{0}$.
In terms of the computational basis, the $\cnot$ gate, which can be rewritten as $\cnot = \ketbra{0}{0}\otimes I + \ketbra{1}{1}\otimes X$, performs a ``NOT'' gate $X$ if the first qubit is set to $\ket{1}$, otherwise does nothing.

For an $n$-qubit system with label $q_j, 1\leq j\leq n$ for each qubit, a $k$-qubit unitary $U \in \CC^{2^k\times 2^k}$ applied to qubits $\bar{q} = q_{j_1},q_{j_2},\ldots,q_{j_k}$ is expanded to the unitary $U_{\bar{q}}\otimes \Id_{\{q_j, 1\leq j\leq n\}\setminus \bar{q}} \in \CC^{2^n\times 2^n}$ that performs a local unitary $U$ on qubits $\bar{q}$ and does nothing with the remaining qubits.
We often write $U_{\bar{q}}$ that omits the identity operator for $U_{\bar{q}}\otimes \Id_{\{q_j, 1\leq j\leq n\}\setminus \bar{q}}$, if there is no ambiguity.
For example, consider a $3$-qubit system with label $j$ for the $j$-th qubit, we write $Z_1$ for $Z\otimes I\otimes I$ and write $X_1Y_3$ for $X\otimes I\otimes Y$.

\subsubsection*{Quantum measurements and observables}
The information about a quantum system has to be acquired by quantum \emph{measurements}.
A measurement on $n$-qubit system is described by a collection $M = \{M_m\}_m$ of $2^n\times 2^n$ complex matrices with the normalization condition  $\sum_m M^{\dagger}_mM_m = \Id_{2^n}$.
When performing it  on a pure state $\ket{\psi}$ and a mixed state $\rho$, the measurement outcome of index $m$ occurs with probabilities $p_m = \bra{\psi}M^{\dagger}_mM_m\ket{\psi}$ and $p_m = \tr(M_m\rho M_m^{\dagger})$, respectively,
and the state after the measurement with outcome $m$ collapses into $\ket{\psi_m} = M_m\ket{\psi}/\sqrt{p_m}$ and $\rho_m = M_m\rho M_m^{\dagger}/p_m$, respectively.
For example, the \emph{computational basis measurement} $\{\ketbra{0}{0}, \ketbra{1}{1}\}$ performed on the state $\ket{+}$ will result in a state $\ket{0}\braket{0}{+}/\sqrt{1/2} = \ket{0}$ with probability $\braket{+}{0}\braket{0}{+} = \abs{\braket{0}{+}}^2 = \frac{1}{2}$ and state $\ket{1}\braket{1}{+}/\sqrt{1/2} = \ket{1}$ with probability $\braket{+}{1}\braket{1}{+} = \abs{\braket{1}{+}}^2 = \frac{1}{2}$.

For an $n$-qubit system with label $q_j,1\leq j\leq n$ for each qubit, the computational basis measurement on a qubit $q \in\{q_1,\ldots,q_n\}$ is given as 
\[\{M_0 = \ketbra[q]{0}{0}\otimes \Id_{\{q_j,1\leq j\leq n\}\setminus q}, M_1 = \ketbra[q]{1}{1}\otimes \Id_{\{q_j,1\leq j\leq n\}\setminus q}\},\]
where we use the same notation as in unitary transformations and also write $\ketbra[q]{j}{j}$ for $\ketbra[q]{j}{j}\otimes \Id_{\{q_j,1\leq j\leq n\}\setminus q}$.

We say a linear operator $O$ is an \emph{observable} if $O^{\dagger} = O$.
The spectral decomposition of an observable $O = \sum_m m P_m$ (a sum over its eigenvalues and corresponding projectors) corresponds to a quantum measurement $\{P_m\}_m$ with outcome $m$ for each projector $P_m$.
For example, Pauli gates $X,Y,Z$ are all observables.
In this paper, for those observables with eigenvalues $\pm 1$, e.g., $Z_1Z_2$, we use Boolean values $0$ and $1$ to indicate the measurement outcomes $(-1)^0 = 1$ and $(-1)^1 = -1$, respectively, to make it convenient in presentations.

\subsection{Stabilizer States and Quantum Error Correction}
\label{sec:qec&stabilizer}

\subsubsection*{Pauli strings and Clifford gates}
An $n$-qubit \emph{Pauli string} $P$ is defined as the tensor products of $n$ Pauli gates with a phase in $\{\pm 1,\pm i\}$, i.e.,
\[P \triangleq i^{k}{P_1}\otimes{P_2}\otimes\cdots\otimes{P_n}, \]
where $k\in\{0,1,2,3\}$, {$P_{\ell}\in\{I,X,Y,Z\}$} for $1\leq \ell\leq n$.
The set of all $n$-qubit Pauli strings forms a group with matrix multiplication.
A \emph{Clifford gate} $V$ is a unitary such that for any Pauli string $P$, $VPV^{\dagger}$ (conjugation by $V$) is still a Pauli string.
In particular, each Pauli gate is a Clifford gate.
Clifford gates have a nice structure; that is, any Clifford gate can be constructed from the three gates: $H$, $S$, and $\cnot$~\cite{gottesman1997stabilizer}, e.g., $I=HH$, $X = HSSH$, $Y = SHSSHSSS$ and $Z = SS$.

\subsubsection*{Stabilizer and stabilizer states}
A state $\ket{\psi}$ is \emph{stabilized} by a gate $U$ if $U\ket{\psi} = \ket{\psi}$, i.e., $\ket{\psi}$ is an eigenvector of $U$ with eigenvalue $1$.
For example, $\ket{+}$ is stabilized by Pauli $X$ gate and $\ket{-}$ is stabilized by $-X$ as $(-X)\ket{-} = \ket{-}$.
For a list of $n$-qubit Pauli strings $P_1, P_2,\ldots, P_k \neq \pm I^{\otimes n}, 1\leq k\leq n$, that are commuting independent\footnote{Each $P_j$ commutes with each other and can't be written as a product of others, i.e., $P_j \not\in \langle \{P_1,\ldots,P_k\}\setminus\{P_j\}\rangle$.} and $P_j^2 \neq -I^{\otimes n}$ for $1\leq j\leq k$, let $\cS = \langle P_1,P_2,\ldots,P_k\rangle$ be the group generated by $\{P_1,P_2,\ldots,P_k\}$.
The subspace of states stabilized by all Pauli strings in $\cS$ is denoted by $V_{\cS}$ and $\cS$ is said to be the \emph{stabilizer} of $V_{\cS}$.
We can check that $\ket{\psi} \in V_{\cS}$ if and only if $P_j\ket{\psi} = \ket{\psi}$ for any $1\leq j\leq k$, thus $V_\cS = \cap_{j=1}^k V_{\langle P_j \rangle}$.
Moreover, $V_\cS$ is a non-trivial subspace with $\dim V_{\cS} = 2^{n-k}$ since each $P_j$ halves the dimension of the space being stabilized.
 
In the case of $k=n$, which implies $\dim V_{\cS} = 1$, we can use the stabilizer $\cS = \langle P_1,P_2,\ldots, P_n\rangle$ to represent the only state $\ket{\psi} \in V_{\cS}$ with global phase ignored;
we say $\cS$ is the stabilizer of $\ket{\psi}$ and $\ket{\psi}$ is a \emph{stabilizer state}.
For example, we have $Z_1\ket{000} = Z_1Z_2\ket{000} = Z_2Z_3\ket{000} = \ket{000}$, so $\langle Z_1, Z_1Z_2, Z_2Z_3\rangle$ is the stabilizer of $\ket{000}$. Also, we have $(-Z_1)\ket{111} = Z_1Z_2\ket{111} = Z_2Z_3\ket{111} = \ket{111}$, and then $\langle -Z_1, Z_1Z_2, Z_2Z_3\rangle$ is the stabilizer of $\ket{111}$.

\subsubsection*{Evolution of stabilizer states by Clifford gates}
Consider an $n$-qubit stabilizer state $\ket{\psi}$ with its stabilizer $\cS = \langle P_1,P_2,\ldots,P_n\rangle$ and a Clifford gate $V$, 
the group $V\cS V^{\dagger} = \langle VP_1V^{\dagger},VP_2V^{\dagger},\ldots,VP_nV^{\dagger}\rangle$ is the stabilizer of the state $V\ket{\psi}$ as $VUV^{\dagger}(V\ket{\psi}) = VU\ket{\psi} = V\ket{\psi}$ for any element $VUV^{\dagger}$ of $VSV^{\dagger}$ with $U\in \cS$.
Thus, the unitary transformation by a Clifford gate $V$ for stabilizer state $\ket{\psi}$ is linked with the conjugation by $V$ for the stabilizer $\cS$ of $\ket{\psi}$.

It is well-known that \emph{stabilizer circuits} consisting of Clifford gates and the computational basis measurements can be formalized into the dynamics of stabilizers, thus can be efficiently simulated on a classical computer by tracking the generators of stabilizers~\cite{aaronson2004improved,gottesman1998heisenberg}. For example, 
consider the stabilizer state $\ket{00}$ with stabilizer $\cS = \langle Z_1, Z_1Z_2\rangle$, an $H_1$ gate followed by a $CNOT_{1,2}$ gate will change the generators $Z_1, Z_1Z_2$ to 
\[\cnot_{1,2}H_1Z_1H_1\cnot_{1,2}= X_1X_2,\qquad \cnot_{1,2}H_1(Z_1Z_2)H_1\cnot_{1,2} = -Y_1Y_2.\]
Then the resulted stabilizer $\cS' = \langle X_1X_2, -Y_1Y_2\rangle$. The stabilizer state of $\cS'$ is $(\ket{00}+\ket{11})/\sqrt{2}$, which is exactly the state $CNOT_{1,2}H_1\ket{00}$.

\subsubsection*{Quantum error correction and stabilizer codes}
The idea of quantum error correction is using a large number $n$ of \emph{physical} qubits to encode  $k$ \emph{logical} qubits to protect logical states against the effects of quantum noise.
A QEC code is identified by the subspace $C$, where $C\subseteq \CC^{2^n}$ and is isomorphic to $\CC^{2^k}$, of all logical states.
A QEC program (decoder) for $C$ consists of two-stage procedure of \emph{error detection} and \emph{recovery}:
\begin{itemize}
    \item In the error detection stage, a list of \emph{syndrome measurements} is performed on physical qubits to detect potential errors, i.e., to check whether the measured state is in the code space $C$.
    The measurement results are called the \emph{error syndromes}.
    \item In the recovery stage, error syndromes are used to decide how to recover the initial state before noise occurs, which can also be said to recover the measured state back to the code space $C$.
\end{itemize}
Quantum noise can be continuous, e.g., arbitrary single qubit unitary occurs at some physical qubits, but fortunately, it has been proved that QEC codes that tolerate a set of discrete errors, e.g., Pauli errors (Pauli operators $\{I, X, Y, Z\}$), can also be used to correct continuous errors~\cite[{Theorem 10.2}]{nielsen2010quantum}.
For QEC codes, the design of stabilizer codes allows us only to consider Pauli $X$ errors (bit-flip errors) and Pauli $Z$ errors (phase-flip errors).

\begin{example}
    [Stabilizer Code]\label{eg:stabilizer_code}
    For a list of $n$-qubit Pauli strings $P_1, P_2, \ldots,P_{n-k}, \bar{Z}_1, \bar{Z}_2, \ldots, \bar{Z}_k \neq \pm I^{\otimes n}$ that are commuting independent with $1\leq k\leq n$ and $P_j^2\neq -I^{\otimes n}, \bar{Z}_j^2\neq -I^{\otimes n}$, an $n$-qubit stabilizer code encodes the logical computational basis state $\ket{x_1,x_2,\ldots,x_k}_L$ of $k$ logical qubits as an $n$-qubit stabilizer state of the stabilizer
    \[\langle (-1)^{x_1}\bar{Z}_1,(-1)^{x_2}\bar{Z}_2,\ldots,(-1)^{x_k}\bar{Z}_k,P_1,P_2,\ldots,P_{n-k},\rangle.\]
    In particular, the space $V_{\langle P_1,P_2,\ldots,P_{n-k}\rangle}$ equals to
    \[\Span\{\ket{x_1,x_2,\ldots,x_k}_L \vert  x_i=0,1\ (i=1,2,\ldots,k)\} \cong \CC^{2^k},\]
    which is the code space of this stabilizer code.
    $P_1,P_2,\ldots,P_{n-k}$ are \emph{stabilizer checks} that take on the role of syndrome measurements, and $\bar{Z}_1,\bar{Z}_2, \ldots, \bar{Z}_k$ are called \emph{logical operators}.
\end{example}

\section{Running Example: Three-Qubit Bit-Flip Code}
\label{sec:example}

For better understanding, in this section, we illustrate the basic idea of our QSE technique through the example of \cref{fig:running_example_decoder}.  
Let us recall the classical bit-flip error,
which flips a bit $0 \rightarrow 1$ and $1 \rightarrow 0$.
A simple way to protect bits of information against this error is the three-repetition code that encodes each bit with three copies of itself: $0 \rightarrow 000$ and $1 \rightarrow 111$.
Suppose a bit-flip error changes an encoded bit string $000$ into $001$,
then through \emph{majority voting}, it can be recovered back to $000$.

\subsubsection*{Quantum error correction}
The program with a two-stage procedure of error detection (Lines \texttt{4}-\texttt{9}) and recovery (Lines \texttt{10}-\texttt{12}) in \cref{fig:running_example_decoder} implements the decoding of the \emph{three-qubit bit-flip code}~\cite[\S{10.1.1}]{nielsen2010quantum}, which is the quantum counterpart of the above three-repetition code.
The code encodes a qubit state $\alpha\ket{0}+\beta\ket{1}$ in three qubits as $\alpha\ket{000}+\beta\ket{111}$ to protect qubits against the quantum bit-flip error: a Pauli $X$ operator at one qubit that takes a qubit state $\alpha\ket{0}+\beta\ket{1}$ to $X(\alpha\ket{0}+\beta\ket{1})=\alpha X\ket{0}+\beta X\ket{1} = \alpha\ket{1}+\beta\ket{0}$.
Suppose an encoded state $\alpha\ket{000}+\beta\ket{111}$ is changed by a quantum bit-flip error into $\alpha\ket{001}+\beta\ket{110}$, which is a superposition of $\ket{001}$ and $\ket{110}$.
To recover this state back to $\alpha\ket{000}+\beta\ket{111}$, we need to design some clever quantum measurements (syndrome measurements) so that we can extract useful information for both $\ket{001}$ and $\ket{110}$ without destroying the superposition state.
Fortunately, there exist two such measurements, the observables $Z_1Z_2$ and $Z_2Z_3$, which are executed in Lines \texttt{4}-\texttt{9} of \cref{fig:running_example_decoder}.
Here, the $\cnot$ gates serve to disentangle the measured qubits, thereby safeguarding the system's information during measurement.
$Z_1Z_2$ checks whether the first and second qubits are the same, and $Z_2Z_3$ checks whether the second and third qubits are the same.
The measurement outcomes of $Z_1Z_2$ and $Z_2Z_3$ on state $\alpha\ket{001}+\beta\ket{110}$ will tell us that the first and second qubit are the same, while the second and third qubit are not the same. Thus we can conclude that the bit-flip error occurs at the third qubit, and then we can perform an $X$ gate at the third qubit to recover $\alpha\ket{001}+\beta\ket{110}$ back to $\alpha\ket{000}+\beta\ket{111}$.


\subsubsection*{Verification of QEC programs} Our goal is to verify the above statement automatically. 
More generally, given a QEC code and its QEC program with error detection and recovery stages, how can we verify that this program can correct all errors allowed by the QEC code for any state in the code space, or provide some useful information for debugging/testing this program if there are some bugs inside it.

\subsubsection*{Our solution --- symbolic execution}

\begin{figure}[ht]
    \centering
    
    \scalebox{.9}{\tikzinput{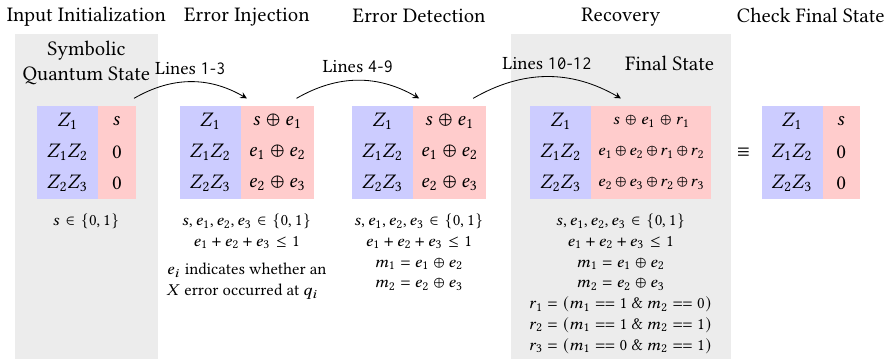}{figures/running_example_SE.tex}}

    \vspace*{-2mm}
    
    \caption{Illustration of symbolic execution for the quantum program in \cref{fig:running_example_decoder}.}\label{fig:running_example}
\end{figure}

To solve this problem, we develop a symbolic execution method for quantum programs.
By introducing a \emph{symbolic quantum state} that can (partially) express possible input quantum states, the quantum program is then ``run'' on this symbolic quantum state.
For the program in \cref{fig:running_example_decoder}, its symbolic execution is illustrated in~\cref{fig:running_example}. More explicitly:
\begin{itemize}\item In the input initialization stage, the symbolic quantum state is the stabilizer $\langle (-1)^sZ_1, Z_1Z_2,$ $Z_2Z_3\rangle$,
a list of Pauli operator strings $(Z_1, Z_1Z_2, Z_2Z_3)$ (the light blue part in~\cref{fig:running_example}) with their phases $(s, 0, 0)$ (the light red part in~\cref{fig:running_example}) containing symbolic variable $s$, that represents the basis $\{\ket{000}, \ket{111}\}$ of the code space. This is because $\ket{000}$ is the stabilizer state of $\langle (-1)^0Z_1, Z_1Z_2,$ $Z_2Z_3\rangle$ and $\ket{111}$ is the stabilizer state of $\langle (-1)^1Z_1, Z_1Z_2,$ $Z_2Z_3\rangle$.
\item In the error injection stage, instead of enumerating concrete values of $e_i$, we consider $e_i$ as symbolic variables over $\{0,1\}$ and define $X^{e_1}, X^{e_2}, X^{e_3}$ as symbolic $X$ errors. The constraint $e_1+e_2+e_3 \leq 1$ indicates that at most one $X$ error occurs, thus covering all allowed errors.
Define $\bar{X}$ as $X_1^{e_1}X_2^{e_2}X_3^{e_3}$. It can be verified that $\bar{X}(Z_1)\bar{X}^{\dagger} = (-1)^{s\oplus e_1}Z_1$, $\bar{X}(Z_1Z_2)\bar{X}^{\dagger} = (-1)^{e_1\oplus e_2}Z_1Z_2$ and $\bar{X}(Z_2Z_3)\bar{X}^{\dagger} = (-1)^{e_2\oplus e_3}Z_2Z_3$.
Thus, as shown in \cref{fig:running_example}, these symbolic $X$ errors make the symbolic quantum state into $\langle (-1)^{s\oplus e_1}Z_1, (-1)^{e_1\oplus e_2}Z_1Z_2, (-1)^{e_2\oplus e_3}Z_2Z_3\rangle$ with constraint $e_1+e_2+e_3 \leq 1$.
\item Then, in the error detection stage, we deduce that $(-1)^{e_1\oplus e_2}Z_1Z_2$ and $(-1)^{e_2\oplus e_3}$ belongs to the stabilizer (group)  $\langle (-1)^{s\oplus e_1}Z_1, (-1)^{e_1\oplus e_2}Z_1Z_2, (-1)^{e_2\oplus e_3}Z_2Z_3\rangle$. Consequently, according to the stabilizer formalism~\cite[{Chapter 10.5}]{nielsen2010quantum}, the measurement outcomes for $Z_1Z_2$ and $Z_2Z_3$ correspond to \emph{symbolic expressions} $m_1 = e_1\oplus e_2$ and $m_2 = e_2\oplus e_3$, respectively.
\item Finally, in the recovery stage, the phase part of the symbolic quantum state becomes a list of symbolic expressions of $e_i, 1\leq i\leq 3$ and  $r_j, 1\leq j\leq 3$ with $r_j$ associating to boolean expressions of conditionals.
This final state in~\cref{fig:running_example} under its constraints (including $e_1+e_2+e_3 \leq 1$), can be verified by an SMT solver to be equivalent to the initial symbolic quantum state. That is, the program with error detection and recovery in \cref{fig:running_example} can correct at least one ($\leq 1$) $X$ error for the symbolic quantum state $\langle (-1)^sZ_1, Z_1Z_2,$ $Z_2Z_3\rangle$.
\end{itemize}
In the same way, the program with error detection and recovery in \cref{fig:running_example} can also correct at least one $X$ error for the symbolic quantum state $\langle (-1)^sX_1X_2X_3, Z_1Z_2, Z_2Z_3\rangle$.
Hence, according to our Theorem~\ref{thm:adequacy}, we can deduce that any state in the code space is satisfied.

\section{Quantum Symbolic Execution}
\label{sec:qse}
The successful application of symbolic execution to our running example (\cref{fig:running_example}) motivates us to develop a framework for the symbolic execution of quantum programs to be applied in formal analysis and verification of other QEC programs. We first introduce a simple quantum programming language that is suitable for describing QEC programs (\S\ref{sec:qse:lang}).
Then, by introducing symbolic quantum states (\S\ref{sec:symbolic_state}), we extend the framework of standard symbolic execution for classical programs to quantum programs (\S\ref{sec:qse:qse}).

\subsection{OpenQASM-like Language}
\label{sec:qse:lang}
We choose to extend IBM's OpenQASM2~\cite{cross2021openqasm}, an imperative quantum programming language supporting classical flow control, with external calls (oracle calls) of classical functions.

\subsubsection*{Syntax}
The set $\qprogs$ of quantum programs is defined by the following syntax:
    \begin{equation*}
        \begin{aligned}
            S \Coloneqq{}& \assign{e}{x} \mid \assign{F(\bar{x})}{\bar{y}}
            \mid{}
            \qut{U}{\bar{q}} \mid \qmeasure{q}{c} 
            \mid{}
            S_1;S_2 \mid \qqif{b}{S_1}{S_2}
        \end{aligned}
    \end{equation*}
    
As in classic imperative languages, the assignment $x\coloneqq e$ evaluates the expression $e$ and assigns the result to the classical variable $x$.
The statement $\assign{F(\bar{x})}{\bar{y}}$ executes an external call of a classical function $F$, an $m$-ary function with $n$ outputs, that accepts the values of a list of input variables $\bar{x} = x_1,x_2,\ldots,x_m$ and assigns the outputs $F(\bar{x})$ to a list of classical variables $\bar{y} = y_1,y_2,\ldots,y_n$ in turn.
This statement is introduced to wrap some classical algorithms used in QEC programs, while independent of quantum variables.
For example, the \emph{Blossom} algorithm~\cite{edmonds1965paths,kolmogorov2009blossom} for minimum weight perfect matching (MWPM) is used extensively in decoding surface codes (2D topological codes)~\cite{dennis2002topological, fowler2012surface}.
Moreover, using such algorithms like the Blossom algorithm is usually by calling existing implementations; thus, we only need to care about the conditions that the input and output meet.

There are two quantum constructs involved with qubit variables.
The unitary transformation $\qut{U}{\bar{q}}$ performs a unitary $U$ on a list of qubit variables $\bar{q} = q_1,q_2,\ldots,q_n$.
For example, $\qut{X}{q}$ performs a Pauli $X$ operator on the qubit $q$ and $\qut{\cnot}{q_1,q_2}$ performs a $\cnot$ operator on qubits $q_1,q_2$.
The measurement statement $\qmeasure{q}{c}$ measures the qubit variable $q$ in the computational basis and assigns the measurement outcome to the classical variable $c$.

Control flow is implemented by sequential composition $S_1;S_2$, which executes subprograms $S_1$ and $S_2$ sequentially; and conditional statement $\qqif{b}{S_1}{S_2}$, which executes the subprogram $S_1$ if the Boolean expression $b$ is evaluated to be \texttt{true}, and otherwise executes the subprogram $S_2$.
For convenience, we use the notation $\qif{b}{S_1}$ that omits the \texttt{false} branch and does nothing when the Boolean expression $b$ is evaluated to be \texttt{false}.

\subsubsection*{Operational semantics}
A \emph{classical-quantum configuration} is a triple $\cfg{S}{\sigma}{\rho}$,
where $S$ is a quantum program or the termination symbol $\terminate$, $\sigma$ is a classical state, and $\rho$ is a density operator denoting the state of qubit variables in $S$.
A probabilistic transition is denoted by $\cfg{S}{\sigma}{\rho} \overto{p} \cfg{S'}{\sigma'}{\rho'}$, where the label $p$ is the transition probability and omitted whenever $p=1$. Then following the semantics of quantum \textbf{while}-language with classical variables given in~\cite{ying2010flowchart},
the operational semantics of our language is defined as a \emph{probabilistic transition relation} $\overto{p}$ between classical-quantum configurations by the transition rules presented in~\cref{fig:operational}.

\begin{figure}[ht]
    \begin{gather*}
        \text{(As)}\ \cfg{\assign{e}{x}}{\sigma}{\rho} \to \cfg{\terminate}{\sigma[\sigma(e)/x]}{\rho} \qquad \text{(EC)}\ \cfg{\assign{F(\bar{x})}{\bar{y}}}{\sigma}{\rho} \to \cfg{\terminate}{\sigma[F(\sigma(\bar{x}))/\bar{y}]}{\rho} \\
        \text{(UT)}\ \cfg{\qut{U}{\bar{q}}}{\sigma}{\rho} \to \cfg{\terminate}{\sigma}{U_{\bar{q}}\rho U_{\bar{q}}^{\dagger}} \ \begin{matrix}
            \text{(M0)}\ \cfg{\qmeasure{q}{c}}{\sigma}{\rho} \overto{\tr(\ketbra[q]{0}{0}\rho)} \cfg{\terminate}{\sigma[0/c]}{\frac{\ketbra[q]{0}{0}\rho\ketbra[q]{0}{0}}{\tr(\ketbra[q]{0}{0}\rho)}} \\
            \text{(M1)}\ \cfg{\qmeasure{q}{c}}{\sigma}{\rho} \overto{\tr(\ketbra[q]{1}{1}\rho)} \cfg{\terminate}{\sigma[1/c]}{\frac{\ketbra[q]{1}{1}\rho\ketbra[q]{1}{1}}{\tr(\ketbra[q]{1}{1}\rho)}}
        \end{matrix} \\
        \text{(SC)}\ \inference{\cfg{S_1}{\sigma}{\rho} \overto{p} \cfg{S_1'}{\sigma'}{\rho'}}{\cfg{S_1;S_2}{\sigma}{\rho} \overto{p} \cfg{S_1';S_2}{\sigma'}{\rho'}} \qquad \begin{matrix}
            \text{(CT)}\ \inference{\sigma \models b}{\cfg{\qqif{b}{S_1}{S_2}}{\sigma}{\rho} \to \cfg{S_1}{\sigma}{\rho}} \\
            \text{(CF)}\ \inference{\sigma \models \neg b}{\cfg{\qqif{b}{S_1}{S_2}}{\sigma}{\rho} \to \cfg{S_2}{\sigma}{\rho}}
        \end{matrix}
    \end{gather*}
    \caption{Transition rules for operational semantics. For external call $\assign{F(\bar{x})}{\bar{y}}$ with $\bar{x} = x_1,x_2,\ldots,x_m$ and $\bar{y} = y_1,y_2,\ldots, y_n$, we use $\sigma[F(\sigma(\bar{x}))/\bar{y}]$ as a shorthand for $\sigma[f_1/y_1,f_2/y_2,\ldots,f_n/y_n]$ provided by $(f_1, f_2, \ldots, f_n) = F(\sigma(x_1), \sigma(x_2), \ldots, \sigma(x_m))$. For sequential composition $S_1;S_2$, we make convention that $\terminate;S_2 = S_2$.}\label{fig:operational}
\end{figure}

The rules (As), (SC), (EC), (CT), and (CF) are the same as in classical or probabilistic programming.
In (UT), quantum state $\rho$ is transformed into $U\rho U^{\dagger}$ by the unitary transformation $U$.
In (M0) and (M1), the computational basis measurement at qubit $q$ brings two probabilistic branches with probability $\tr(\ketbra[q]{0}{0}\rho)$ and $\tr(\ketbra[q]{1}{1}\rho)$, respectively, that transform quantum state $\rho$ into $\frac{\ketbra[q]{0}{0}\rho\ketbra[q]{0}{0}}{\tr(\ketbra[q]{0}{0}\rho)}$ and $\frac{\ketbra[q]{1}{1}\rho\ketbra[q]{1}{1}}{\tr(\ketbra[q]{1}{1}\rho)}$, respectively, and assign the corresponding measurement outcomes $0$ and $1$, respectively, to classical variable $c$.

For convenience, we also write $\cfg{S}{\sigma}{\rho} \overto[*]{p} \cfg{S'}{\sigma'}{\rho'}$ if there is a sequence of transitions
\[\cfg{S}{\sigma}{\rho} \overto{p_1} \cfg{S_1}{\sigma_1}{\rho_1} \overto{p_2} \cdots  \overto{p_n} \cfg{S'}{\sigma'}{\rho'}\]
such that $n\geq 0$ and $p = \prod_{j=1}^np_j$.

\subsection{Symbolic Quantum States}
\label{sec:symbolic_state}
Just like using symbols to represent a range of values and building expressions over constants and symbols, we want to build a term that depends on some symbols, and when concrete values substitute these symbols, this term becomes a quantum state---density operator.
We begin with a few examples of density operators with parameters.

\begin{example}
    [Handling quantum states of single-qubit]\label{eg:qubit_state}
    Consider a qubit system with state space $\cH = \Span\{\ket{0}, \ket{1}\}$.
    A pure qubit state can be written as 
        \[\ket{\psi(\theta, \phi)} = \cos\frac{\theta}{2}\ket{0}+e^{i\phi}\sin\frac{\theta}{2}\ket{1}, \theta\in [0,\pi], \phi\in[0,2\pi],\]
        with a global phase ignored~\cite[Bloch Sphere]{nielsen2010quantum}. 
        Then, its corresponding density operator is $\psi(\theta, \phi) = \ketbra{\psi(\theta, \phi)}{\psi(\theta, \phi)}$. By replacing variable $\theta$ with a symbol $s_0$ over $[0,\pi]$ and variable $\phi$ with a symbol $s_1$ over $[0,2\pi]$, we would obtain a symbolic quantum state $\psi(s_0,s_1)$.
\end{example}

\begin{example}[Handling quantum states of $n$-qubit]\label{eg:all_states}
Consider an $n$-qubit system with state space $\cH = \Span\{\ket{0},\ket{1},\ldots,\ket{2^n-1}\}$.
\begin{enumerate}
    \item For any $(\alpha_0,\alpha_1,\ldots,\alpha_{2^n-1})\in \CC^{2^n}$, let $\ket{g(\alpha_0,\alpha_1,\ldots,\alpha_{2^n-1})} = \nicefrac{1}{\sqrt{\sum_{k=0}^{2^n-1}\abs{\alpha_k}^2}}\sum_{i=0}^{2^n-1}\alpha_i\ket{i}$ parameterize all $n$-qubit pure states.
    Then, with each $\alpha_j$ replaced by a symbol $s_j$ over $\CC$, $g(s_0,s_1,\ldots,s_{2^n-1})$ would be the symbolic quantum state we want. 
    \item Moreover, we can choose $\ket{g_{j}(\alpha_{j,0},\alpha_{j,1},\ldots,\alpha_{j,2^n-1})} = \nicefrac{1}{\sqrt{\sum_{k=0}^{2^n-1}\abs{\alpha_{j,k}}^2}}\sum_{i=0}^{2^n-1}\alpha_{j,i}\ket{i}$ and $(\beta_0,$ $\beta_1,\ldots,\beta_{2^n-1}) \in \CC^{2^n}$ to define \[
        h(\beta_0,\beta_1,\ldots,\beta_{2^n-1},\ldots,\alpha_{j,k},\ldots) = 
    \sum_{j=0}^{2^n-1}\frac{\abs{t_j}^2}{\sum_{j=0}^{2^n-1}\abs{t_j}^2}g_j(\alpha_{j,0},\alpha_{j,1},\ldots,\alpha_{j,2^n-1}),\]
    which parameterizes all $n$-qubit mixed states.
    Then, with $t_j, s_{i,j}$ being symbols over $\CC$, $h(\beta_0,$ $\beta_1,\ldots,\beta_{2^n-1},\ldots,\alpha_{j,k},\ldots)$ would be a symbolic quantum state for $n$-qubit mixed states.
    \end{enumerate}
\end{example}

Building on \cref{eg:qubit_state,eg:all_states}, we introduce the general form of symbolic quantum states.

\begin{definition}
    [Symbolic Quantum State]\label{def:sqs}
    A symbolic quantum state $\srho$ is an $n$-ary map 
    \[G:C_1\times C_2\times\cdots\times C_n \to \cD\]
    with each argument substituted by a symbol over its domain, i.e.,
    \[\srho \triangleq G(s_1, s_2, \ldots, s_n),\]
    where $n$ is a positive integer, for each $1\leq i\leq n$, $C_i$ is a subset of complex numbers, and $s_i$ is the symbol for the value ranging  over $C_i$, and $\cD$ is the set of quantum states under consideration, e.g., the space of $m$-qubit mixed states if the program contains $m$ qubits.
\end{definition}

\subsubsection*{Quantum operations over symbolic quantum states}
To use symbolic quantum states to capture the quantum states in quantum programs, we need to further define unitary transformations and quantum measurements over them.

\begin{definition}
    [Quantum operations over symbolic quantum states]\label{def:operations}
    Let $\srho$ be a symbolic quantum state for qubit variables $q_1, q_2, \ldots, q_n$. Then:
    \begin{itemize}
        \item For a  unitary statement $\qut{U}{\bar{q}}$ with qubits $\bar{q} \subseteq \{q_1, q_2, \ldots, q_n\}$, the unitary transformation function $ut$ is defined by
        \begin{equation}\label{eq:ut_function}
            ut(U, \bar{q}, \srho) = U_{\bar{q}} \srho U_{\bar{q}}^{\dagger},
        \end{equation}
        where $U_{\bar{q}} \triangleq U_{\bar{q}}\otimes \Id_{\{q_1, q_2, \ldots, q_n\}\setminus \bar{q}}$ is the operation  that performs a local unitary $U$ on qubits $\bar{q}$ and does nothing with the remaining qubits.
        \item For a  measurement statement $\qmeasure{q}{c}$ with $q \in \{q_1,q_2,\ldots,q_n\}$, the measurement function $m$ is defined by
        \begin{equation}\label{eq:m_function}
            m(q, \srho) = \biggl(s, \tr\bigl(\ketbra[q]{s}{s}\srho\bigr), \frac{\ketbra[q]{s}{s}\srho\ketbra[q]{s}{s}}{\tr(\ketbra[q]{s}{s}\srho)}\biggr)
        \end{equation}
        with $s$ a newly introduced symbol over $\{0,1\}$, where $\ketbra[q]{s}{s} \triangleq \ketbra[q]{s}{s}\otimes \Id_{\{q_1, q_2, \ldots, q_n\}\setminus \{q\}}$ is the measurement operator that act on qubit $q$ as $\ketbra{s}{s}$ and does nothing with the remaining qubits.
        The output is a triple of a newly introduced symbol $s$ ranging over measurement outcomes $\{0,1\}$, a symbolic expression $p(s) = \tr\bigl(\ketbra[q]{s}{s}\srho\bigr)$ for the probability of obtaining outcome $s$ when performing measurement at qubit $q$ of state $\srho$ and a symbolic quantum state $\srho'(s) = \frac{\ketbra[q]{s}{s}\srho\ketbra[q]{s}{s}}{\tr(\ketbra[q]{s}{s}\srho)}$ that represents the state after this measurement.
    \end{itemize}
\end{definition}

\subsubsection*{Efficient representations} According to our \cref{def:sqs}, there exist infinitely many symbolic quantum states $F:C_1\times C_2\times \cdots \times C_n \to \cD$.
However, in practical applications, we must carefully consider which type of $F$ can be efficiently implemented on a  classical computer.
If we trivially employ the symbolic quantum state in \cref{eg:all_states}, we would encounter several challenges:
\begin{itemize}
    \item \textbf{Exponential complexity}: the number of symbols required, i.e., the dimension of the state space, grows exponentially with respect to the number of qubits, leading to computational complexity beyond what classical computers can handle.
    \item \textbf{Lack of powerful solvers}: without specialized solvers, addressing the problems arising from QSE would become challenging.
\end{itemize}
On the other hand, the development of SMT solvers promotes the application of classical SE.
Similarly, we must make the symbolic quantum state take advantage of the existing solvers.

Fortunately, for QEC programs, we can construct suitable symbolic quantum states with efficient representations---symbolic stabilizer states, which we will discuss later in \S\ref{sec:symbolic_stabilizer}.

\subsection{Quantum Symbolic Execution}
\label{sec:qse:qse}

Following the idea of classical SE, we come up with an SE framework for quantum programs (QSE) that maintains a symbolic configuration $\scfg{S}{\ssigma}{\srho}{P}{\varphi}$, where:
\begin{itemize}
    \item $S$ is the quantum program to be executed.
    \item $\ssigma$ is the symbolic classical state, which is the same as the symbolic state in classical SE, and maps classical variables of the program to symbolic expressions over constants and symbolic values.
    \item $\srho$ is the \emph{symbolic quantum state}, equipped with a unitary transformation function $ut(U, \bar{q}, \srho)$ for statement $\qut{U}{\bar{q}}$ and a measurement function $m(q, \srho)$ for statement $\qmeasure{q}{c}$, for qubit variables of the program. 
    \item $P$ is a set that records the symbolic probabilities corresponding to the outcomes of quantum measurements in the program. At the beginning of the execution, $P = \emptyset$.
    \item $\varphi$ is the path condition, i.e., a conjunctive formula that expresses a set of assumptions on the symbols due to external calls and branches taken in an execution path of the quantum program. 
\end{itemize}

As a natural extension to the operational semantics of quantum programs in~\cref{fig:operational}, the QSE changes the symbolic configuration according to the QSE rules presented in~\cref{fig:qse_rules}.
We also write $\scfg{S}{\ssigma}{\srho}{P}{\varphi} \to^{*} \scfg{S'}{\ssigma'}{\srho'}{P'}{\varphi'}$ if there is a sequence of transitions 
\[\scfg{S}{\ssigma}{\srho}{P}{\varphi} \to \scfg{S_1}{\ssigma_1}{\srho_1}{P_1}{\varphi_1} 
    \to \cdots 
    \to \scfg{S_{{n}}}{\ssigma_{{n}}}{\srho_{{n}}}{P_{{n}}}{\varphi_{{n}}}
    {=} \scfg{S'}{\ssigma'}{\srho'}{P'}{\varphi'}\]
such that $n\geq 0$.

We assume that the function $F$ in an external call has a logical formula $C_{F}(\bar{x}, \bar{y})$ asserting the condition that input and output must meet, i.e.,  $C_{F}(\bar{x}, F(\bar{x})) \equiv \texttt{true}$.
The existence of this condition is obvious, e.g., we can choose $C_F(\bar{x}, \bar{y}) = (\bar{y} == F(\bar{x}))$, but such a choice may not be efficient to (symbolically) compute for any $F$.

For measurement statements, our defined function $m$ in~\cref{def:operations} uses a symbol to represent the measurement outcomes, avoiding branches like conditional statements.
Since a lot of measurements are commonly required in QEC, this trick has the advantage of not facing the exponential complexity over the number of measurements.

\begin{figure}
    \begin{align*}
        \text{(S-As)} &\ \scfg{\assign{e}{x}}{\ssigma}{\srho}{P}{\varphi} \to \scfg{\terminate}{\ssigma[\ssigma(e)/x]}{\srho}{P}{\varphi} \\
        \text{(S-EC)} &\ \scfg{\assign{F(\bar{x})}{\bar{y}}}{\ssigma}{\srho}{P}{\varphi} \to \scfg{\terminate}{\ssigma[\bar{s}_{\bar{y}}/\bar{y}]}{\srho}{P}{\varphi \land C_{F}(\ssigma(\bar{x}), \bar{s}_{\bar{y}})} \\
        \text{(S-UT)}&\ \scfg{\qut{U}{\bar{q}}}{\ssigma}{\srho}{P}{\varphi} \to  \scfg{\terminate}{\ssigma}{ut(U, \bar{q}, \srho)}{P}{\varphi} \\
        \text{(S-M)}&\ \scfg{\qmeasure{q}{c}}{\ssigma}{\srho}{P}{\varphi} \to \scfg{\terminate}{\ssigma[s/c]}{\srho'(s)}{P\cup\{(s,p(s))\}}{\varphi} \\
        & \qquad \text{ where } (s,p(s),\srho'(s)) = m(q,\srho) \\
        \text{(S-SC)}&\ \inference{\scfg{S_1}{\ssigma}{\srho}{P}{\varphi} \to \scfg{S_1'}{\ssigma'}{\srho'}{P'}{\varphi'}}{\scfg{S_1;S_2}{\ssigma}{\srho}{P}{\varphi} \to \scfg{S_1';S_2}{\ssigma'}{\srho'}{P'}{\varphi'}} \\
        \text{(S-CT)}&\ \scfg{\qqif{b}{S_1}{S_2}}{\ssigma}{\srho}{P}{\varphi} \to \scfg{S_1}{\ssigma}{\srho}{P}{\varphi \land \ssigma(b)} \\[-2pt]
        \text{(S-CF)}&\ \scfg{\qqif{b}{S_1}{S_2}}{\ssigma}{\srho}{P}{\varphi} \to \scfg{S_2}{\ssigma}{\srho}{P}{\varphi \land \neg\ssigma(b)}
    \end{align*}
    \caption{Symbolic execution rules for quantum programs. For external call $\assign{F(\bar{x})}{\bar{y}}$, $C_{F}(\ssigma(\bar{x}), \bar{s}_{\bar{y}})$ is the condition that the input and output should satisfy for $F$, i.e., $C_{F}(\bar{x}, F(\bar{x})) \equiv \texttt{true}$ for any concrete values $\bar{x}$, where $\bar{s}_{\bar{y}} = s_{y_1}, s_{y_2}, \ldots, s_{y_n}$ is a list of newly introduced symbols for output variables $\bar{y} = y_1,y_2,\ldots,y_n$.}\label{fig:qse_rules}
\end{figure}

\section{Soundness Theorem}
\label{sec:soundness}
In this section, we present the soundness theorem for QSE (namely, the QSE rules in~\cref{fig:qse_rules}) with respect  to the operational semantics.
The proof for this section is located in \ifARXIV{\cref{app:soundess_prf}}\else{Appendix~C of the extended version~\cite{fang2023symbolic}}\fi.

To precisely state the theorem, we need the notion of instantiation for symbolic quantum states. 

\begin{definition}[Instantiation]\label{def:satisfaction}Let $(\sigma, \rho)$ be 
    a pair of classical state and quantum state, and let $(\ssigma, \srho)$ be a pair of symbolic classical state and symbolic quantum state. Then we say that $(\sigma, \rho)$ is an instantiation of $(\ssigma, \srho)$ under a path condition $\varphi$, written $$(\sigma, \rho) \models_{\varphi} (\ssigma, \srho)$$ if there is a \emph{valuation} $V$ that assigns every symbol in $\ssigma, \srho$ and $\varphi$ with concrete values such that $V(\varphi) = \texttt{true}$, $\sigma = V(\ssigma)$ and $\rho = V(\srho)$.
    In particular, if $\varphi \equiv  \texttt{true}$, we simply write $(\sigma, \rho) \models (\ssigma, \srho)$.
\end{definition}

Our soundness theorem relies on the following two key lemmas about functions $ut$ and $m$ for symbolic quantum states introduced in~\cref{def:operations}.

\begin{restatable}[Correctness of function $ut$]{lemma}{correctnessut}\label{lem:ut}
    For a unitary transformation statement $\qut{U}{\bar{q}}$, a path condition $\varphi$, a pair $(\sigma, \rho)$ of classical state and quantum state, and a pair $(\ssigma, \srho)$ of symbolic classical state and symbolic quantum state, if $\cfg{\qut{U}{\bar{q}}}{\sigma}{\rho} \to \cfg{\terminate}{\sigma'}{\rho'}$ and $(\sigma, \rho) \models_{\varphi} (\ssigma, \srho)$, then
    \[(\sigma', \rho') \models_{\varphi} (\ssigma, ut(U, \bar{q}, \srho)).\]
\end{restatable}

\begin{restatable}[Correctness of function $m$]{lemma}{correctnessm}\label{lem:m}
    For a measurement statement $\qmeasure{q}{c}$, a path condition $\varphi$, a pair $(\sigma, \rho)$ of classical state and quantum state, and a pair $(\ssigma, \srho)$ of symbolic classical state and symbolic quantum state, if $\cfg{\qmeasure{q}{c}}{\sigma}{\rho} \overto{p} \cfg{\terminate}{\sigma'}{\rho'}$ and $(\sigma, \rho) \models_{\varphi} (\ssigma, \srho)$, then 
    \[(\sigma', \rho') \models_{\varphi} (\ssigma[s/c], \srho'(s)),\]
    where $(s, p(s), \srho'(s)) = m(q, \srho')$.
\end{restatable}

Now we are ready to present the soundness theorem.

\begin{restatable}[Soundness of QSE]{theorem}{soundness}\label{thm:soundess}
    For any quantum program $S$, any pair $(\sigma, \rho)$ of classical state and quantum state, any pair $(\ssigma, \srho)$ of symbolic classical state and symbolic quantum state, any set $P$ of symbolic probabilities and any path condition $\varphi$, if 
    \begin{gather*}
        (\sigma, \rho) \models_{\varphi} (\ssigma, \srho),\quad \cfg{S}{\sigma}{\rho} \overto{p} \cfg{S'}{\sigma'}{\rho'},\quad 
        \scfg{S}{\ssigma}{\srho}{P}{\varphi} \to \scfg{S'}{\ssigma'}{\srho'}{P'}{\varphi'}
    \end{gather*}
    then $(\sigma', \rho') \models_{\varphi'} (\ssigma', \srho')$. Moreover, there is a valuation $V$ such that $V(\ssigma) = \sigma, V(\srho) = \rho, V(\varphi) = \<true>, V(\ssigma') = \sigma', V(\srho') = \rho', V(\varphi') = \<true>$ and
    if $P' = P\cup\{(s, p(s)\}$, then $p = V(p(s))$.
\end{restatable}

\section{Symbolic Stabilizer States}
\label{sec:symbolic_stabilizer}

As we briefly discussed in~\S\ref{sec:symbolic_state}, a general symbolic quantum state cannot be represented in an efficient way. In this section, we propose a special class of symbolic quantum state with efficient representations,  called \emph{symbolic stabilizer state}. This class of symbolic quantum states is what we need in the analysis and verification of QEC programs. Then we extend this symbolic representation to unitaries and measurements so that we have an efficient QSE framework for our application in the analysis and verification of stabilizer codes.

\subsection{Symbolic Stabilizer States}
\label{sec:symbolic_stabilizer0}
First, we present an efficient symbolic representation of stabilizer states by  introducing symbols into the phases of Pauli strings as follows.

\begin{definition}
    [Symbolic stabilizer state]\label{def:sss}
    For any commuting independent set $\{P_1,P_2,\ldots,P_{n} \mid p_{j}\neq \pm I^{\otimes n}, P_{j}^2 \neq -I^{\otimes n}\}$ of $n$-qubit Pauli strings with size $n$ and $n$ Boolean functions $f_1,f_2,\ldots,f_n$ over $m$ Boolean variables,
    let $\ket{\psi(b_1,b_2,\ldots,b_m)}$ denote a stabilizer state of 
    \[\langle (-1)^{f_1(b_1,\ldots,b_m)}P_1,\ldots,(-1)^{f_n(b_1,\ldots,b_m)}P_n\rangle.\]
    A symbolic stabilizer state is defined as a symbolic quantum state
    \[\srho = \psi(s_1,s_2,\ldots,s_m) = \ketbra{\psi(s_1, s_2, \ldots, s_m)}{\psi(s_1, s_2, \ldots, s_m)}\]
    with $s_1,s_2,\ldots,s_m$ being symbols over Boolean values.
\end{definition}

Since the global phase of $\ket{\psi(b_1,b_2,\ldots,b_m)}$ is canceled out in $\psi(b_1,b_2,\ldots,b_m)$, the stabilizer $\langle (-1)^{f_1(b_1,\ldots,b_m)}P_1,\ldots,(-1)^{f_n(b_1,\ldots,b_m)}P_n\rangle$ corresponds to a unique density operator $\psi(b_1,b_2,\ldots,b_m)$.
Therefore, we slightly abuse the notation and also write the density operator
\[\psi(b_1,b_2,\ldots,b_m) = \langle (-1)^{f_1(b_1,\ldots,b_m)}P_1,\ldots,(-1)^{f_n(b_1,\ldots,b_m)}P_n\rangle,\]
and the symbolic stabilizer state
\[\srho = \langle (-1)^{f_1(s_1,\ldots,s_m)}P_1,\ldots,(-1)^{f_n(s_1,\ldots,s_m)}P_n\rangle.\]

\begin{example}
    In our running example (\cref{fig:running_example}), the initial symbolic stabilizer state is \[\langle (-1)^sZ_1, (-1)^0Z_1Z_2, (-1)^0Z_2Z_3\rangle\] with $s$ being a Boolean symbol. It represents a set of quantum states $\{\ket{000}, \ket{111}\}$ as $\ket{000}$ and $\ket{111}$ are the stabilizer states of $\langle Z_1, Z_1Z_2, Z_2Z_3\rangle$  and  $\langle -Z_1, Z_1Z_2, Z_2Z_3\rangle$, respectively.
\end{example}

\subsection{Quantum Operations over Symbolic Stabilizer States}\label{sec:sss_operations}
Now we can extend the symbolic representation for stabilizer states to unitaries and measurements. Indeed, the unitary transformation function $ut$ and measurement function $m$ defined in \S\ref{sec:symbolic_state} can be naturally lifted to functions on stabilizers by the stabilizer formalism~\cite[Chapter 10.5]{nielsen2010quantum}.
The definitions are as follows, and their correctness for \Cref{lem:m,lem:ut} is available in \ifARXIV{\cref{app:adeqacy_prf}}\else{Appendix~D of the extended version~\cite{fang2023symbolic}}\fi.

\begin{definition}
    [Quantum operations over symbolic stabilizer states]\label{def:qo_sss}
    Let \[\srho = \langle (-1)^{f_1(s_1,\ldots,s_m)}P_1,\ldots,(-1)^{f_n(s_1,\ldots,s_m)}P_n\rangle\] be a symbolic stabilizer state for qubit variables $q_1,q_2,\ldots, q_n$. Then
    \begin{itemize}
        \item For a Clifford unitary statement $\qut{V}{\bar{q}}$ with qubits $\bar{q}\subseteq \{q_1,\ldots,q_n\}$, the unitary transformation function $ut$ is defined as $ut(V, \bar{q}, \srho)= $
        \[ \langle (-1)^{f_1(s_1,\ldots,s_m)}V_{\bar{q}}P_1V_{\bar{q}}^{\dagger},\ldots,(-1)^{f_n(s_1,\ldots,s_m)}V_{\bar{q}}P_n V_{\bar{q}}^{\dagger}\rangle.\]
        \item For a measurement statement $\qmeasure{q}{c}$ with $q \in \{q_1,q_2,\ldots,q_n\}$, the definition of measurement function $m$ is divided into two cases: \begin{enumerate}
            \item $Z_q$ commutes with all $P_j,1\leq j\leq n$.
            In this case, there exist a Boolean value $b$ and a list of indexes $1\leq j_1\leq j_2\leq \ldots\leq j_k\leq n$ with $1\leq k\leq n$ such that $Z_q = (-1)^bP_{j_1}P_{j_2}\cdots P_{j_k}$. The measurement does not change the state and implies a determinate outcome~\cite[\S10.5.3]{nielsen2010quantum}.
            The function $m$ is defined as $m(q, \srho) = $
            \[(b\oplus f_{j_1}(s_1,\ldots,s_m)\oplus \cdots\oplus  f_{j_k}(s_1,\ldots,s_m),\{\}, \srho).\]
            \item $Z_q$ anti-commutes  one or more of $P_j$. In this case, without loss of generality, we can assume $P_1$ anti-commutes with $Z_q$ and $P_2,\ldots,P_n$ commute with $Z_q$\footnote{If there is other $P_j, j\geq 2$ anti-commutes with $Z_q$, we can replace the $(-1)^{f_j(s_1,\ldots,s_m)}P_j$ by $(-1)^{f_1(s_1,\ldots,s_m)+f_j(s_1,\ldots,s_m)}P_1P_j$ in the generating set of $\srho$, and it will result in the same $\srho$. Then $Z_q$ only commute with $P_1P_j$.}.
            Then, the function $m$ is defined as $m(q, \srho) =$
            \[ \hspace{-5mm}\Bigl(s, \frac{1}{2}, \langle (-1)^{s}Z_q, (-1)^{f_2(s_1,\ldots,s_m)}P_2,\ldots,(-1)^{f_n(s_1,\ldots,s_m)}P_n\rangle\Bigr),\]
            where $s$ is a newly introduced Boolean symbol for the measurement outcome.
        \end{enumerate}
    \end{itemize}
\end{definition}

\begin{example}
    In our running example~\cref{fig:running_example}, measuring $Z_1Z_2$ (Lines \texttt{4}-\texttt{6} in~\cref{fig:running_example_decoder}) is implemented by $\qut{\cnot}{q_1,q_2}; \qmeasure{q_2}{m_1}; \qut{\cnot}{q_1,q_2}$. It changes the symbolic stabilizer state $\srho_2 = \langle (-1)^{s\oplus e_1}Z_1, (-1)^{e_1\oplus e_2}Z_1Z_2, (-1)^{e_2\oplus e_3}Z_2Z_3\rangle$ as follows: \begin{enumerate}
        \item $\qut{\cnot}{q_1,q_2}$ transforms $\srho_2$ to ${\srho}_3$ \begin{align*}
           ={}&\scalebox{.9}{$\langle (-1)^{s\oplus e_1}\cnot_{1,2}Z_1\cnot_{1,2}, (-1)^{e_1\oplus e_2}\cnot_{1,2}(Z_1Z_2)\cnot_{1,2}, (-1)^{e_2\oplus e_3}\cnot_{1,2}(Z_2Z_3)\cnot_{1,2}\rangle$} \\
            ={}& \langle (-1)^{s\oplus e_1}Z_1, (-1)^{e_1\oplus e_2}Z_2, (-1)^{e_2\oplus e_3}Z_1Z_2Z_3\rangle.
        \end{align*}
        \item Then $\qmeasure{q_2}{m_1}$ performs measurement $Z_2$ on $\srho_3$. Since $Z_2$ commutes with all Pauli strings in ${\srho}_3$, we follow the case 1 in~\cref{def:qo_sss} to get $Z_2 = (-1)^0Z_2$, thus $m(q_2, \srho_3) = (e_1\oplus e_2, \{\}, \srho_3)$ leading to outcome $m_1 = e_1\oplus e_2$ and the quantum state remains $\srho_3$.
        \item Finally, $CNOT\ q_1,q_2$ transforms $\srho_3$ back to $\srho_2$.
    \end{enumerate}
\end{example}

\subsubsection*{Efficient implementation}
The two functions $ut$ and $m$ defined above require three subroutines:
\begin{itemize}
    \item A subroutine that computes $VPV^{\dagger}$ for a Pauli string $P$ and a Clifford gate $V$.
    \item A subroutine that determines whether two Pauli strings $P_1$ and $P_2$ are commutable or not.
    \item A subroutine that produce a Boolean value $b$ and a list of indexes $1\leq j_1,j_2,\ldots,j_k\leq n$ with $1\leq k\leq n$ such that $(-1)^bZ_q = P_{j_1}P_{j_2}\cdots P_{j_k}$ for a qubit $q$ and a stabilizer $\langle P_1,P_2,\ldots, p_n\rangle$ if $Z_q$ commutes with all $P_j$.
\end{itemize}
These subroutines are standard in stabilizer formalism, and there are already efficient implementations~\cite{aaronson2004improved,gidney2021stim} for them.
We provide a detailed discussion of the implementation in \ifARXIV{\cref{sec:implementation}}\else{Appendix~B of the extended version~\cite{fang2023symbolic}}\fi.

\subsection{Adequacy}
\label{sec:adequacy}
In this subsection, we further show that our symbolic representation is adequate for our target application in the analysis and verification of stabilizer codes.
The proof for this section is located in \ifARXIV{\cref{app:adeqacy_prf}}\else{Appendix~D of the extended version~\cite{fang2023symbolic}}\fi.

To answer whether a decoder of a stabilizer code is correct, we need to verify that for any quantum state $\rho$ in the code space as the input state, the program $S$ consisting of error injection (inject errors allowed by the stabilizer code) and this decoder will output the same state $\rho$.
Thus, we need to verify whether the input quantum state and the output quantum state are the same.

Formally, let $S$ be a quantum program that only uses Clifford gates and $\sigma$ be a classical state. For $1\leq k\leq n$, assume that an $n$-qubit stabilizer code encoding $k$ logical qubits with syndrome operators $P_1,P_2,\ldots,P_{n-k}$ and logical operators $L_1,L_2,\ldots,L_k$. Let $\ket{x_1,x_2,\ldots,x_k}_L$ with $x_j\in\{0,1\}$ denoting the logical computational basis state of the code as in~\Cref{eg:stabilizer_code}, and $\bar{H} = H_{L,1}\otimes H_{L,2}\otimes\cdots\otimes H_{L,k}$ with $H_{L,j}$ being the logical Hardamard gate, which is a Clifford gate, for the $j$-th logical qubit.
    It is clear that  states $\ket{x_1,x_2,\ldots,x_k}_L$ and $\bar{H}\ket{x_1,x_2,\ldots,x_k}_L$ are both stabilizer states. We write $\mathcal{T}$ for the set of (density operators of) these states with $x_j\in\{0,1\} (1\leq j\leq k)$. Then we have:

\begin{restatable}{lemma}{adequacylem}\label{lem:adequacy}
    If for any $\rho\in\mathcal{T}$, it holds that 
\begin{equation}\label{transition-stab}\cfg{S}{\sigma}{\rho} \overto{p}^{*} \cfg{\terminate}{\sigma'}{\rho'} \text{ with $p>0$ implies $\rho' = \rho$}\end{equation}    
    then (\ref{transition-stab}) holds for all quantum state $\rho$ in the code space.
\end{restatable}

\begin{restatable}[Adequacy of symbolic stabilizer states]{theorem}{adequacythm}\label{thm:adequacy}
    Let:\begin{enumerate} \item $\srho_1$ be the symbolic stabilizer state $\langle P_1,\ldots,P_{n-k},(-1)^{s_1}L_1,\ldots,(-1)^{s_k}L_k\rangle$; and \item $\srho_2$ be the symbolic stabilizer state $\langle \bar{H}P_1\bar{H},\ldots,\bar{H}P_{n-k}\bar{H},(-1)^{s_1}\bar{H}L_1\bar{H},\ldots,(-1)^{s_k}\bar{H}L_k\bar{H}\rangle$ with $s_j$ symbols over Boolean values. \item $\ssigma$ be a classical symbolic state that does not contains $s_j,1\leq j\leq k$. \end{enumerate}
     If for $\srho=\srho_1,\srho_2$, $$\scfg{S}{\ssigma}{\srho}{\emptyset}{\texttt{true}} \to^{*} \scfg{\terminate}{\ssigma'}{\srho'}{P'}{\varphi'}\text{ implies } \varphi'\models\srho' = \srho, $$
     then for any quantum state $\rho$ in the code space and any valuation $V$,
     $$\cfg{S}{V(\ssigma)}{\rho} \overto{p}^{*} \cfg{\terminate}{\sigma'}{\rho'} \text{ with $p>0$ implies $\rho' = \rho$}.$$
\end{restatable}

This theorem tells us that we only need to check two symbolic stabilizer states to verify that a program's initial and final quantum states are the same.

\section{Conditional Application of Pauli Gates without Forking}\label{sec:conditional}
In classical SE, a branch of conditional triggers a path fork and update to the path constraints as rules (S-CT) and (S-CF) in \cref{fig:qse_rules}.
However, there may be many conditional statements in QEC programs, e.g.,
the conditionals in our running example (Lines \texttt{10}-\texttt{12} in \cref{fig:running_example_decoder}) are designed for each physical qubit; thus, when the number of physical qubits is large, the number of conditional statements required is also large.
Then, we will face the problem of path explosion.
Fortunately, with a newly designed rule (see \cref{eq:qse_rule_spg} later), our symbolic stabilizer states can solve this problem.

\subsubsection*{Symbolic Pauli gates}
We observe that to correct the $X$ or $Z$ error, QEC programs usually use a conditional statement to determine whether to apply the $X$ or $Z$ gate in the recovery stage, e.g.,  the line \texttt{10} in \cref{fig:running_example_decoder}:
\[\qif{m_1==1 \ \&\  m_2==0}{\qut{X}{q_1}}\]
This application of $X$ and $Z$ gates inspires us to define the following symbolic Pauli gates.

\begin{definition}
    [Symbolic Pauli gates]
    For a symbolic Boolean expression $e$, a Pauli gate $\tau \in \{X, Y, Z\}$, and a qubit $q$, we define the symbolic Pauli gate $\tau^{e}$ at qubit $q$ as the conditional statement $\qif{e}{\qut{\tau}{q}}$; that is,
    $\tau^e$ applies $\tau$ to qubit $q$ if $e$ is evaluated to be \texttt{true} and $I$ if $e$ is evaluated to be \texttt{false}.
    For convenience, we also add a new statement
    \[\qut{\tau[e]}{q} \equiv \qif{e}{\qut{\tau}{q}}\]
    for symbolic Pauli gate $\tau^e$.
\end{definition}

\subsubsection*{QSE rule for symbolic Pauli gates}
Our symbolic stabilizer states are very friendly to symbolic Pauli gates.
For any Pauli string \[P = (i)^kP_1\otimes P_2\otimes\cdots\otimes P_n,\]
conjugation by a symbolic Pauli gate $\tau^e$ at qubit $q$ transforms it into 
\[(i)^kP_1\otimes\cdots\otimes(\tau^eP_q(\tau^\dagger)^e)\otimes \cdots\otimes P_n = (-1)^{e\cdot[\tau\neq P_q]}P,\]
where $[\tau\neq P_q] = 1$ if $\tau \neq P_q$, otherwise $[\tau\neq P_q] = 0$. We see that $\tau^e$ only changes the phase of the Pauli string. Thus, for symbolic stabilizer states, the unitary transformation by a symbolic Pauli gate $\tau^e$ can be rewritten as
\begin{align*}
    & ut(\tau^e, q, \langle (-1)^{f_1(s_1,\ldots,s_m)}P_1,\ldots,(-1)^{f_n(s_1,\ldots,s_m)}P_n\rangle) \\
    ={}& \langle (-1)^{f_1(s_1,\ldots,s_m)\oplus e[\tau\neq P_{1,q}]}P_1,\ldots,(-1)^{f_n(s_1,\ldots,s_m)\oplus e[\tau\neq P_{n,q}]}P_n\rangle.
\end{align*}
where the $q$-th Pauli gate of $P_k$ is $P_{k,q}$.
Then, the QSE rule for $\qut{\tau[e]}{q}$ can be redefined as
\begin{align}
\label{eq:qse_rule_spg}
    \text{(S-SP)} &\ \scfg{\qut{\tau[e]}{q}}{\ssigma}{\srho}{P}{\varphi} \to \scfg{\terminate}{\ssigma}{ut(\tau^{\ssigma(e)}, q, \srho)}{P}{\varphi}
\end{align}

\subsubsection*{Conditional without forking}
The rule (S-SP) provides a simpler and more effective way to handle the statement $\qut{\tau[e]}{q}$ than rules (S-CT) and (S-CF) for conditionals in~\cref{fig:qse_rules}.
For example, consider the program
\[\qut{X[m_1==0]}{q_1};\qut{X[m_2==0]}{q_2};\qut{X[m_3==0]}{q_3},\]
which is also 
\[\qif{m_1==0}{\qut{X}{q_1}};\qif{m_2==0}{\qut{X}{q_2}};\qif{m_3==0}{\qut{X}{q_3}},\]
for a symbolic configuration, the rule (S-SP) will eventually produce one symbolic configuration; however, rules (S-CT) and (S-CF) will produce $8$ symbolic configurations.

\subsubsection*{Inserting symbolic Pauli errors through symbolic
Pauli gates}
The error injection in our running example (see \cref{fig:running_example_decoder}) can be done by the program
\[\qut{X[e_1==1]}{q_1}; \qut{X[e_2==1]}{q_2}; \qut{X[e_3==1]}{q_3}\]
with constraints $e_1+e_2+e_3 \leq 1,0\leq e_1,e_2,e_3\leq 1$, which characterizes all possible $X$ errors with at most $1$ location of physical qubits.
More generally, the following program
\begin{equation}\label{eg:symbolic_errors}
    \qut{X[e_1==1]}{q_1}; \qut{X[e_2==1]}{q_2}; \ldots ; \qut{X[e_n==1]}{q_n}
\end{equation}
with constraints $e_1+e_2+\cdots+e_n \leq d$, $0\leq e_j\leq 1, 1\leq j\leq n$, captures all possible $X$ errors at most $d$ locations of $n$ physical qubits.
By doing this, we efficiently overcome the challenge of dealing with a large number of samples of possible Pauli errors, which grows exponentially as $\binom{n}{d}$.

\section{Experimental Evaluation}
\label{sec:evaluation}
We implemented our  general QSE framework together with the special support of symbolic stabilizer states in a prototype tool called \QuantumSE{}
as a Julia~\cite{bezanson2017julia} package.
For an initialized symbolic configuration $\scfg{S}{\ssigma}{\srho}{P}{\varphi}$, \QuantumSE{} constructs a set of the terminal configurations  $\scfg{\terminate}{\ssigma'}{\srho'}{P'}{\varphi'}$ s.t. $\scfg{S}{\ssigma}{\srho}{P}{\varphi} \to^{*} \scfg{\terminate}{\ssigma'}{\srho'}{P'}{\varphi'}$ by the QSE rules in \cref{fig:qse_rules}. 
Then, for an assertion that we are interested in at a terminal configuration  $\scfg{\terminate}{\ssigma'}{\srho'}{P'}{\varphi'}$, e.g., $\varphi\land\varphi' \models \srho' = \srho$ for QEC decoders, \QuantumSE{} will call the Bitwuzla SMT solver~\cite{DBLP:conf/cav/NiemetzP23} to prove it\footnote{We choose Bitwuzla as it is a winner of the competition in the quantifier-free bit-vector logic category in SMT-COMP 2022 (see \url{https://smt-comp.github.io/2022/results/qf-bitvec-single-query}).}.
If the SMT solver is unable to prove it, the SMT solver will provide a counterexample, from which \QuantumSE{} will report that the program $S$ is buggy and generate the test case.

To evaluate \QuantumSE{}, we consider the following research questions (RQs):
\begin{itemize}
    \item \textbf{RQ1.} \textit{Is \QuantumSE{} scalable at finding bugs in QEC programs over different kinds of QEC codes?}
    \item \textbf{RQ2.} \textit{Can \QuantumSE{} outperform other tools?}
    \item \textbf{RQ3.} \textit{What factors affect the performance of \QuantumSE{}?}
\end{itemize}
All our experiments are carried out on a desktop with Intel(R) Core(TM) i7-9700 CPU @3.00GHz and 16G of RAM, running Ubuntu 22.04.2 LTS. 

\subsection{RQ1: Finding Bugs in QEC Programs}\label{sec:rq1}
To address RQ1, we selected three representative QEC codes, i.e., quantum repetition codes, Kitaev's toric codes~\cite{kitaev1997quantum,kitaev2003fault}, and quantum Tanner codes~\cite{leverrier2022quantum}.
The repetition codes and toric codes belong to surface codes~\cite{dennis2002topological}, a variant of which was implemented in Google's recent QEC experiment~\cite{Acharya2022SuppressingQE} (with 72 physical qubits).
The quantum Tanner codes, which follow the recent major breakthrough~\cite{panteleev2022asymptotically} of QEC codes, are pretty complicated.
\cref{fig:qec_repetition} shows one of the QEC programs evaluated in this experiment. It has a strange bug that it will fail to correct errors when $X$ errors occur at qubits $q_1,q_2,\ldots,q_{\lfloor \frac{n-1}{2}\rfloor}$.
A detailed description of the QEC codes and programs used in this experiment is available in \ifARXIV{\cref{sec:cases}}\else{Appendix~A of the extended version~\cite{fang2023symbolic}}\fi.

\begin{figure}[ht]

    \begin{subfigure}[b]{.485\linewidth}
        \centering
    \scalebox{.64}{\tikz{\node[fill=black!15,inner sep=0pt]{$\setlength{\arraycolsep}{1pt}
    \def\arraystretch{1.25}\begin{array}{ll}
        \qut{CNOT}{q_1,q_2}; \qmeasure{q_2}{s_1}; \qut{CNOT}{q_1,q_2}; & \text{// measure } Z_1Z_2\\
        \qut{CNOT}{q_2,q_3}; \qmeasure{q_3}{s_2}; \qut{CNOT}{q_2,q_3}; & \text{// measure } Z_2Z_3 \\
        \cdots \\
        \qut{CNOT}{q_{n-1},q_{n}}; \qmeasure{q_n}{s_{n-1}}; \qut{CNOT}{q_{n-1},q_{n}}; & \text{// measure } Z_{n-1}Z_n\\
        \qut{CNOT}{q_n,q_1}; \qmeasure{q_1}{s_n}; \qut{CNOT}{q_n,q_1}; & \text{// measure } Z_nZ_1\\
        \assign{\text{MWPM}(s_1,s_2,\ldots,s_n)}{r_1,r_2,\ldots,r_n}; & \text{// call MWPM}\\
        \qut{X[r_1==1]}{q_1}; & \text{// apply $X$ if $r_1=1$}\\
        \qut{X[r_2==1]}{q_2}; & \text{// apply $X$ if $r_2=1$} \\
        \cdots \\
        \qut{X[r_n==1]}{q_n}; & \text{// apply $X$ if $r_n=1$}\\
        \qut{X[r_1*r_2*\cdots *r_{\lfloor \frac{n-1}{2}\rfloor}==1]}{q_1} & \text{// a strange bug!}
    \end{array}$};}}

    \vspace*{-1mm}

    \caption{A QEC program with a strange bug for quantum repetition codes. The quantum repetition code with $n$ physical qubits has stabilizer checks $Z_1Z_2, Z_2Z_3,\ldots,Z_{n-1}Z_n$. After measuring these checks, we additionally measure the observable $Z_nZ_1$ for the convenience of calling MWPM.}\label{fig:qec_repetition}
    \end{subfigure}
    \hfill
    \begin{subfigure}[b]{.5\linewidth}
        \centering
    \tikzinput{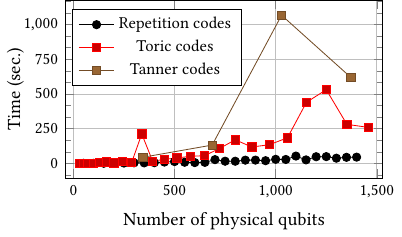}{figures/finding_bugs.tex}

    \vspace*{-3mm}

    \caption{Performance results of finding bugs in QEC programs by \QuantumSE{}. The runtime here is mainly consumed by the SMT solver, and its irregular trend is also due to performance anomalies of the solver in some special cases. See \S\ref{sec:rq3} for detailed analysis.}\label{fig:rq1}
    \end{subfigure}

    \vspace*{-2mm}

    \caption{An example of QEC programs evaluated in RQ1 and the performance results of RQ1.}
\end{figure}

\paragraph{Results.} Based on the construction of these QEC codes, we evaluated \QuantumSE{} on repetition codes with $50 j$ physical qubits for $1\leq j\leq 28$, toric codes with $2d^2$ physical qubits for $4\leq d\leq 27$, and quantum Tanner codes with $343k$  physical qubits for $1\leq k\leq 4$.
\cref{fig:rq1} shows the running times for \QuantumSE{} to find bugs in the QEC programs for these different sizes of QEC codes.
We can see that the running time does not tend to increase dramatically with the number of qubits.
\QuantumSE{} can analyze QEC programs with over 1000 qubits in a short period of time, exceeding the current physical experiments that use about 280 qubits.
Thus, we believe that QSE will be useful in future QEC experiments for debugging QEC programs.

\subsection{RQ2: Comparing with Other Tools}\label{sec:rq2}

\subsubsection*{Comparing with other SE tools}
As we will discuss in related work (see \S\ref{sec:related_work_2}), none of the existing work on SE can handle QEC programs.
Nevertheless, since \texttt{SymQV}~\cite{bauer2023symqv} provides an easy-to-use implementation, we compare \QuantumSE{} with it on a simple class of circuits that only contains $H$ and $X$ gates (see \cref{fig:simple_circuit}).
We use \texttt{SymQV} and \QuantumSE{} respectively to check whether the input states and output states of \cref{fig:simple_circuit} are equivalent.
Since \texttt{SymQV} uses Z3 as the SMT solver, we also use Z3 in this comparison.
\begin{figure}[ht]
    \centering
    \vspace*{-3mm}

    \begin{subfigure}[b]{.4\linewidth}
        \centering
        \scalebox{0.8}{\begin{quantikz}[row sep=2mm]
            \qw & \gate{H} & \gate{X} & \gate{X} & \gate{H} & \qw \\
            \qw & \gate{H} & \gate{X} & \gate{X} & \gate{H} & \qw \\[-1mm]
             \vdots & \vdots & \vdots & \vdots & \vdots & \vdots \\
            \qw & \gate{H} & \gate{X} & \gate{X} & \gate{H} & \qw
        \end{quantikz}}
        \caption{A simple class of quantum circuits that contains four layers of $H, X, X, H$ gates, which are canceled out into identity.}\label{fig:simple_circuit}
    \end{subfigure}
    \hfil
    \begin{subfigure}[b]{.59\linewidth}
    \centering
    \tikzinput{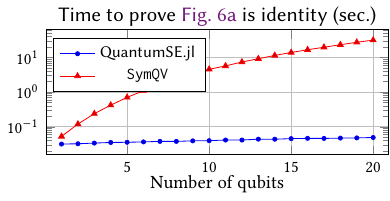}{figures/comp_simple_circuit.tex}
    \vspace{-2mm}
    \caption{Performance results of \texttt{SymQV} and \QuantumSE{}.}\label{fig:comp_simple_circuit}
    \end{subfigure}

    \vspace*{-2mm}
    \caption{Comparison between \texttt{SymQV} and \QuantumSE{}. \QuantumSE{}'s runtime is fast, while \texttt{SymQV}'s is slower even though we set up \texttt{SymQV} with product state and overapproximation.}
    \label{fig:comp}
    \vspace*{-4.5mm}
\end{figure}

\paragraph{Results.} The runtime of \texttt{SymQV} and \QuantumSE{} are presented in \cref{fig:comp_simple_circuit}.
The main reason for the results in \cref{fig:comp_simple_circuit} is that \QuantumSE{} is based on stabilizers, while \texttt{SymQV} is based on quantum states, which has a much larger dimension that grows exponentially with the number of qubits.
But the strange thing is that we have set up the product state approximation for \texttt{SymQV}, which should not take too long as in \cref{fig:comp_simple_circuit}.
We believe there is room for optimization in implementing \texttt{SymQV}.

\subsubsection*{Comparing with state-of-the-art simulators}
Several other tools can handle QEC programs (circuits), but since the dimension of the state space of qubits grows exponentially with the number of qubits, the only existing work that can efficiently handle QEC circuits with over 1000 qubits are stabilizer-based simulators (see discussion in \S\ref{sec:related_work_1}).
However, finding bugs with the simulator requires considering all possible errors.
For example, in the QEC program of \cref{fig:qec_repetition} with $n=1000$, there are $\binom{1000}{499}\approx 10^{299}$ possible locations of $499$ $X$ errors, but the bug only occurs when $q_1,q_2,\ldots,q_{499}$ have $X$ errors, which is in one of these possibilities.
\emph{Even though the state-of-the-art stabilizer simulator, Stim~\cite{gidney2021stim}, is really fast at sampling measurement results, it is unrealistic to sample one case out of $10^{299}$ possibilities even with a supercomputer.}

To give Stim a fair shake, we have added the ability to sample stabilizer circuits' measurement results to \QuantumSE{}%
\footnote{%
After symbolic execution, we only need to substitute concrete values for all symbols (with corresponding probabilities) in the symbolic expressions of measurement results to achieve sampling without having to go through the circuit again.
}
and compare it to Stim in terms of sampling functionality.
A sample of a stabilizer circuit refer to the (random) measurement outcomes obtained by running the circuit once.
For samples of a stabilizer circuit, Stim and \QuantumSE{} first initialize a sampler and then use it to generate samples rapidly.
To conduct a comparison, we choose the benchmark of layered random interaction circuits used in Stim~\cite{gidney2021stim}, where Stim outperformed popular simulators such as Qiskit's stabilizer method~\cite{aleksandrowicz2019qiskit}, Cirq's Clifford simulator~\cite{omole2020cirq}, \citet{aaronson2004improved}'s \texttt{chp.c} and GraphSim~\cite{anders2006fast}.

\begin{figure}[ht]
    \centering
    \vspace*{-2mm}

    \tikzinput{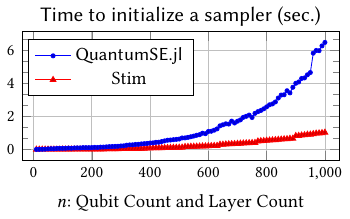}{figures/benchmark_init.tex}
    \hfil
    \tikzinput{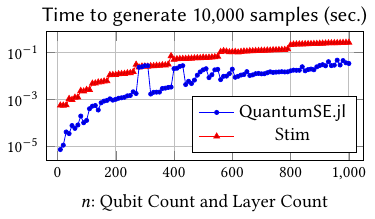}{figures/benchmark_sampling.tex}

    \vspace*{-3mm}

    \caption{Performance results of sampling layered random interaction circuits. Each circuit is made up of $n$ qubits with $n$ layers. Each layer randomly applies an $H$, $S$, and $I$ gate to each qubit, then samples $10$ pairing of the qubits to apply $CNOT$ gate, then samples $5\%$ of the qubits to measure in the computational basis. At the end of the circuit, each qubit is measured in the computational basis.}\label{fig:rq2}

    \vspace*{-4mm}
\end{figure}

\paragraph{Results.}
We present the comparisons of the time for \QuantumSE{} and Stim to initialize a sampler (i.e., the time to analyze the input circuit and create a sampler for generating the measurement results) and the time for \QuantumSE{} and Stim's samplers to generate 10,000 samples in \cref{fig:rq2}.
We can see that for $n > 600$, the sampler generated by \QuantumSE{} spends a shorter time than Stim's in generating samples, which means that \emph{\QuantumSE{} can exhibit a faster sampling speed than the state-of-the-art stabilizer simulator}.
We also notice that \QuantumSE{} spends more time than Stim to initialize the sampler.
The main reason here is that Stim is well optimized in initializing the sampler, while our \QuantumSE{} is not specifically designed for it.

\subsection{RQ3: Analyzing Performance Factors}\label{sec:rq3}
To analyze the factors that affect the performance of \QuantumSE{}, we count the running time of \QuantumSE{} debugging QEC programs in three separate stages:
\begin{itemize}
    \item \textbf{Init}: the initialization stage that prepares initial symbolic configurations for QSE. It takes time mainly to initialize the symbolic quantum states as in \cref{thm:adequacy} and to insert symbolic errors
    as in \cref{eg:symbolic_errors}.
    \item \textbf{QSE}: the quantum symbolic execution stage that performs symbolic execution of the quantum program. The time spent in this stage depends on how \QuantumSE{} maintains the symbolic configuration, especially its symbolic quantum state (symbolic stabilizer state).
    \item \textbf{SMT}: the SMT solver stage that calls an SMT solver to solve the assertions constructed by \QuantumSE{}. The time spent in this stage depends on the performance of the solver.  
\end{itemize}

\begin{figure}[ht]
    \centering
    \vspace*{-3mm}

    \tikzinput{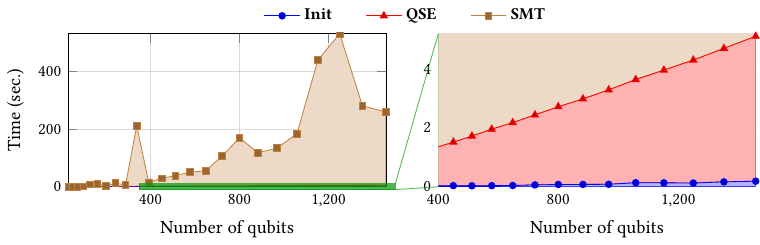}{figures/stack_time.tex}
    \vspace*{-3mm}
    \caption{The stack graph of time spent in three stages.}\label{fig:rq3}

    \vspace*{-4mm}
\end{figure}

\paragraph{Results.}
We chose the QEC programs for Toric codes in RQ1 (\S\ref{sec:rq1}) to analyze because of the nontrivial trend of Toric codes in the time curve (\cref{fig:rq1}) and enough number of data points available.
We show the time spent by \QuantumSE{} in the three stages in a stack graph, \cref{fig:rq3}, which allows us to see that the stage \textbf{SMT} is dominating the running time.
We can see that the time used by the SMT solver does not increase regularly with the number of qubits. We find that this is due to anomalies in the performance of the SMT solver in some special cases\footnote{We have tried many SMT solvers, including Z3, Z3++, Yices2, and cvc5. They also have similar problems, but the instances where the problem occurs are not the same instances that Bitwuzla encounters.}.
For example, the adopted SMT solver Bitwuzla spent $210.3$(s) in the case of $338$ qubits; however, if we change the configuration of Bitwuzla, it can solve it with only $11$(s).
This raises an interesting issue regarding improving the combination of QSE and SMT solvers for further research.
Additionally, we {observe} that the time share of stages \textbf{Init} and \textbf{QSE} is small and the trend of their time curves is smooth.
Thus, these two parts are very scalable and do not restrict the performance of \QuantumSE{}.

\section{Related Work}
\label{sec:related_work}

\subsection{Simulation-based analysis for quantum computing}\label{sec:related_work_1}
With the rapid development of quantum hardware, researchers have recently put a lot of effort into the simulation of quantum computation on  classical computers, and with it comes a wealth of quantum softwares~\cite{omole2020cirq,bergholm2018pennylane,aleksandrowicz2019qiskit,svore2018q,suzuki2021qulacs,luo2020yao}, which help us to test and debug quantum circuits and programs.

We first discuss the \emph{stabilizer-based simulation} that is most relevant to our work. It is well-known that stabilizer circuits can be efficiently simulated~\cite{gottesman1998heisenberg,aaronson2004improved,gidney2021stim,anders2006fast,krastanov}, where \citet{gottesman1998heisenberg}'s tableau algorithm and \citet{aaronson2004improved}'s improved tableau algorithm play a significant role because of the tableau representation of stabilizers and destabilizers.
Based on the tableau representation and algorithms, \citet{rall2019simulation} proposed the Pauli propagation that speeds up the rate of sampling Pauli noises, and it was adopted by Google's Stim~\cite{gidney2021stim}, which was used in Google's recent QEC experiments~\cite{Acharya2022SuppressingQE}.
However, \emph{even though Stim has excellent performance, it still cannot handle all possible errors for QEC programs as we discussed in \S\ref{sec:rq2}}.

It is worth noting that, starting from the tableau representation again, \citet{berent2022towards} proposed an SAT encoding for Clifford circuits and demonstrated the applicability of equivalence checking; \citet{schneider2023sat} used a similar SAT encoding for Clifford circuit synthesis.
However, 
\textit{all of these works only dealt with quantum circuits without mid-circuit measurements and control flows, and thus cannot be directly used in debugging QEC programs considered in this paper}.
{Additionally}, \citet{Rand2021} uses a type system to describe the stabilizer formalism elegantly. This type system provides efficient verification of quantum circuits. However, it needs to deal with all possible errors for QEC programs on a case-by-case analysis like a simulator.

There are also other techniques for simulation of quantum computation, e.g., tensor networks-based simulation~\cite{markov2008simulating,orus2019tensor,pan2022simulation}, (binary) decision diagrams-based simulation~\cite{niemann2015qmdds,vinkhuijzen2021limdd,hong2022tensor,sistla2022cflobdds}.
These simulation techniques are widely used in developing or testing quantum circuits and quantum programs~\cite{burgholzer2023tensor}.
\emph{But most works are limited to manipulating concrete data and do not introduce symbolic expressions like ours.}

\subsection{Symbolic techniques for quantum computing}\label{sec:related_work_2}
Due to the difficulty of simulation and analysis of quantum computation on classical computers,  symbolic techniques have already been employed to improve their efficiency and scalability.   
The current works can be roughly classified into the following two categories: 

\subsubsection*{Symbolic simulation of quantum computation}
Several work~\cite{sistla2023symbolic,sistla2022cflobdds,tsai2021bit} 
used symbolic expressions or symbolic Boolean
functions during the process of simulation of quantum computation to speed up the simulation time; for example, 
\citet{tsai2021bit} built a series of Boolean formulas between {binary decision diagrams} for state evolution to replace matrix-vector multiplication.
{Furthermore}, \citet{chen2023automata} used the tree automata~\cite{comon2008tree} to represent  (set of) quantum states and introduced symbolic update formulae of tree automata for quantum gates.
\citet{10.1145/3445814.3446750} converts the variational quantum circuit into logical formulas, in which the parameters of the variational quantum circuit are temporarily symbolized. Then with concrete values assigned, logical formulas provide an efficient sampling of quantum circuits.
\emph{Although they make use of symbolic expressions, these works focused on simulation, and 
the idea of symbolic execution was not introduced there}.

\subsubsection*{Symbolic execution of quantum circuits}
\citet{carette2023symbolic} used the algebraic normal form of Boolean functions to perform the symbolic execution of Hardamard-Toffoli quantum circuits,  which cannot express QEC.
Another work close to ours is \textsf{symQV}~\cite{bauer2023symqv}, in which 
$n$-qubit quantum states are represented with $2^n$ symbolic complex numbers $\ket{\psi} \coloneqq (\alpha_1,\alpha_2,\ldots,\alpha_{2^n})$ or product state $\ket{\psi} \coloneqq \otimes_{j=1}^n\ket{\psi_j}$, where $\ket{\psi_j}$ is encoded by $4$ symbols.
Their experiments demonstrated their applicability to 24 qubits. \emph{Although the target of \textsf{symQV} is the symbolic execution of quantum programs, it is actually designed for quantum circuits without control flows, and thus \textsf{symQV} cannot be directly used in the verification of QEC programs considered in this paper}.
{In addition}, {Giallar}~\cite{10.1145/3519939.3523431} introduced symbolic representation and execution for quantum circuits with $20$ rewrite rules to check the equivalence of quantum circuits, which also lack support for control flows; Quartz~\cite{10.1145/3519939.3523433} dealt with symbolic quantum circuits too, computing symbolic matrix representations of circuits to discover equivalent circuit transformations by SMT solvers.
However, similar to Giallar, Quartz does not support control flows.

\subsection{Verification and Analysis of Quantum Programs}\label{sec:related_work_3}
There is a rich literature on the verification and analysis of quantum programs.
We briefly discuss three categories of research relevant to our work.

\subsubsection*{Formal verification with program logic}
Based on quantum predicates of observables~\cite{d2006quantum}, \citet{ying2012floyd} proposed the first sound and relatively complete quantum Hoare logic (QHL) for quantum \textbf{while}-programs.
To simplify the verification, \citet{10.1145/3314221.3314584} used projectors as predicates and proposed a variant of QHL.
Since projectors may involve exponential complexity for QEC programs, several works suggested to use stabilizers as predicates and proposed corresponding variants of QHL~\cite{wu2021qecv,Sundaram2022}.
Their approaches share a similar idea with \citet{Rand2021}'s type system and still need to deal with all possible errors for QEC programs on a case-by-case verification.
Beside quantum predicates of observables, Qbricks~\cite{10.1007/978-3-030-72019-3_6} used parameterized path sum representations for quantum circuits and reasoned about the representation using quantifier-free predicate logic. It has good proof automation and deduction rules for reasoning on circuits, but currently lacks support for classical control flows.

\subsubsection*{Formal verification without program logic} The most representative work that does not use program logic is perhaps  QWIRE~\cite{10.1145/3009837.3009894,Rand_2018} and SQIR~\cite{10.1145/3434318,https://doi.org/10.4230/lipics.itp.2021.21}, both of which directly formalize the denotational semantics of quantum programs in terms of density matrices.
In this way, it seems that they may also encounter exponential complexity, when applied for verification and analysis of  QEC programs.

\subsubsection*{Assertion checking and debugging}
In addition to formal verification, an interesting  line of research is assertion checking for analyzing and debugging quantum programs.
Some papers use dynamic assertions~\cite{10.1145/3428218,10.1145/3373376.3378488} to detect bugs at run-time, which seems not suitable for QEC programs’ correctness before deployment.
Others use statistical assertions~\cite {10.1145/3307650.3322213}, which may require repeated simulations and are thus inefficient for the case of a large number of qubits.
To verify assertions more efficiently, \citet{10.1145/3453483.3454061} introduced a novel  idea of abstract interpretation that abstracts concrete quantum states into the intersection of small projections, which are closely related to the definition of stabilizer.
For instance, in their example, the abstract domain for GHZ is a stabilizer.
However, their approach also requires a case-by-case check for all possible errors in QEC programs.

\section{Conclusion and Future Directions}
\label{sec:conclusion}
In this paper, we presented a symbolic execution framework for quantum programs.
Within this framework, we introduced symbolic stabilizer states, which facilitate the efficient analysis of QEC programs.
We also developed a prototype tool based on this framework and demonstrated its effectiveness and efficiency. Issues for future research include: \begin{enumerate}[leftmargin=5mm]
\item \textit{Incorporating our QSE with the recent technique of probabilistic symbolic execution~\cite{susag2022symbolic,gehr2016psi}}: This will enable us to handle the analysis of random errors in QEC.
Specifically, it can be used to analyze the performance of QEC programs against random errors when we have verified their correctness against adversarial errors.
\item \textit{Improving the combination of QSE and SMT solvers}:
A key point here is to optimize the SMT solver for the specific assertions posed by QSE.
\item \textit{Exploring more applications of QSE and symbolic stabilizer states}: 
A possible application is the sampling task we have done in \S\ref{sec:rq2}, promising to perform better with the existing engineering efforts. Another is to use symbolic stabilizer states for assertions, which may make existing work possible to handle all adversarial errors and all logical states in QEC without enumeration.
\end{enumerate}

\begin{acks}
  We thank anonymous reviewers for helpful comments and suggestions that improved this paper.
  We are grateful to Craig Gidney for pointing out our previous inappropriate use of Stim.
  This work was partly supported by the \grantsponsor{}{National Key R\&D Program of China}{} under Grant No.~\grantnum{}{2023YFA1009403}.
\end{acks}

\section*{Data Availability Statement}
Our code is available at \url{https://github.com/njuwfang/QuantumSE.jl}.\\
An evaluated artifact~\cite{artifact} is available at \url{https://doi.org/10.5281/zenodo.10781381}.

\bibliographystyle{ACM-Reference-Format}
\bibliography{ref}

\appendix

\section{Case Studies}
\label{sec:cases}
This section introduces the case studies of QEC codes used in our evaluation (\S\ref{sec:evaluation}).
The selection of our case studies is to demonstrate  the effectiveness and efficiency of our proposed QSE with a prototype tool (\S\ref{sec:implementation}).
For a QEC code and a QEC program for this code, we utilize QSE to analyze whether this program can correct quantum errors.

Before introducing our case studies, we recall some notations used for a QEC code $C$.
\begin{itemize}
    \item The minimum distance $d_X$ (resp. $d_Z$) of $C$ is the minimum number of different locations of $X$ (resp. $Z$) operators that transform a state $\ket{\psi}\in C$ into another state $\ket{\phi} \neq \ket{\psi}$ in $C$.
    Therefore, $C$ can tolerate $X$ (resp. $Z$) errors at least $\lfloor \frac{d_X-1}{2}\rfloor$ (resp. $\lfloor \frac{d_Z-1}{2}\rfloor$) locations of physical qubits.
    \item $[[n,k,d]]$ describes the parameters of a QEC code: number of physical qubits $n$, number of logical qubits $k$ and minimum distance $d = \min\{d_X, d_Z\}$.
    We also write $[[n ,k, d_X=\_]]$ and $[[n ,k, d_Z=\_]]$ if we only focus on the minimum distance $d_X$ and $d_Z$, respectively.
\end{itemize}

\subsection{Repetition Codes}\label{sec:repetition_code}
Repetition codes may be the simplest class of QEC codes.
Through an encoding map
\[\ket{0}_L \mapsto \ket{0}^{\otimes n} = \ket{\underbrace{00\cdots 0}_{n \text{ times } 0}}, \qquad \ket{1}_L \mapsto \ket{1}^{\otimes n} = \ket{\underbrace{11\cdots 1}_{n \text{ times } 1}},\]
repetition code $C_n = \{\alpha \ket{0}^{\otimes n}+ \beta\ket{1}^{\otimes n}\}$ encodes $1$ logical qubit with $n$ physical qubits.
The $d_X$ of $C_n$ is $n$, while the $d_Z$ of $C_n$ is $1$.
Hence, the repetition code $C_n$ can correct $\lfloor \frac{n-1}{2}\rfloor$ locations of $X$ errors, but can't correct for any $Z$ errors.

\begin{figure}[ht]
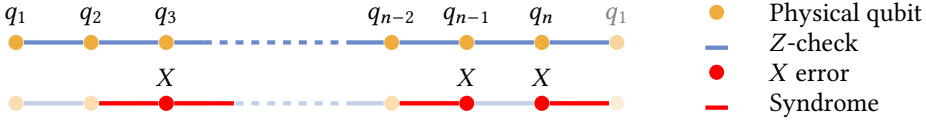

    \centering
    \[
        \mathtikz[b]{
        \foreach \i in {1,2,3}{
            \node[circle, fill=qubitcolor, inner sep=2pt,label={$q_{\i}$}] (p\i) at (\i, 0) {};
            \node[circle, fill=qubitcolor, inner sep=2pt,opacity=0.4] (e\i) at (\i, -0.8) {};
        };
        \foreach \i/\l in {4/n-2,5/n-1,6/n}{
            \node[circle, fill=qubitcolor, inner sep=2pt, label={$q_{\l}$}] (p\i) at (\i+2, 0) {};
            \node[circle, fill=qubitcolor, inner sep=2pt,opacity=0.4] (e\i) at (\i+2, -0.8) {}; 
        };
        \foreach \i/\j in {1/2,2/3,4/5,5/6}{
            \path[draw,zcheckcolor!70, ultra thick] (p\i) -- (p\j);
            \path[draw,zcheckcolor!70, ultra thick, opacity=0.4] (e\i) -- (e\j);
        }
        \node[circle, fill=qubitcolor, inner sep=2pt, label={[opacity=0.5]:$q_{1}$}, opacity=0.5] (p11) at (7+2, 0) {};
        \node[circle, fill=qubitcolor, inner sep=2pt, opacity=0.2] (e11) at (7+2, -0.8) {};
        \foreach \n/\o in {p/1,e/0.4}{
            \path[draw,zcheckcolor!70, ultra thick, opacity=\o] (\n6) -- (\n11);
            \path[draw,zcheckcolor!70, ultra thick,opacity=\o] (\n3) -- +(0.5,0) edge[dashed] +(2.5,0);
            \path[draw,,zcheckcolor!70, ultra thick,opacity=\o] (\n4) -- +(-0.5,0);
        }
        \foreach \i in {3,5,6}{
            \node[circle, fill=red, inner sep=2pt, label={$X$}] at (e\i) {};
        }
        \foreach \i/\j in {2/3, 4/5, 6/11}{
            \path[draw, red, ultra thick] (e\i) -- (e\j);
        }
        \path[draw,red, ultra thick] (e3) -- +(0.9,0);
        }
        \quad \def\arraystretch{0.9}
        \begin{array}[b]{cl}
            \mathtikz{\node[circle, fill=qubitcolor, inner sep=2pt] {};} & \text{Physical qubit}\\
            \mathtikz{\path[draw,zcheckcolor!70, ultra thick] (0,0) -- (0.3,0);} & Z\text{-check} \\
            \mathtikz{\node[circle, fill=red, inner sep=2pt] {};} & X\text{ error} \\
            \mathtikz{\path[draw,red, ultra thick] (0,0) -- (0.3,0);} & \text{Syndrome}
        \end{array}
    \]
    \caption{\textbf{The repetition code.} Physical qubits, represented by meat-brown points, lie on a 1-D line with a period boundary (a circle). Syndrome measurements, represented by blue-gray edges, are observables $ZZ$ between each pair of adjacent points.}\label{fig:repetition_code}
\end{figure}

The repetition code $C_n$ is also a one-dimensional version of the surface code~\cite{dennis2002topological}.
For the example shown in \cref{fig:repetition_code} for $C_n$, $n$ physical qubits $q_1,q_2,\ldots, q_n$ lie on a line with period boundary and each edge between adjacent points $q_{j}, q_{j+1}$ defines a syndrome measurement, observable $Z_{j}Z_{j+1}$, called $Z$-check.
Consider the stabilizer of these $Z$-checks,
$\cS_n =\langle Z_1Z_2,Z_2Z_3,\ldots, Z_{n-1}Z_{n}\rangle$,
we can find that the space $V_{\cS_n}$ of states stabilized by $\cS_n$ is exactly the repetition code $C_n$.
This representation by stabilizer also provides an explicit decoding strategy: as demonstrated in \cref{fig:repetition_code}, $X$ errors (the red points in \cref{fig:repetition_code}) will affect the measurement outcome (syndrome) of $Z$-check if there is only one X error occurred at its vertices (boundary), see the red lines in \cref{fig:repetition_code}.
With these syndromes in hand, we can use the MWPM algorithm to find a minimum number of locations of $X$ operators that can produce the same syndromes.
Then, applying $X$ gates at these locations of qubits, we are done with the decoding process.

\begin{figure*}[ht]
    \centering
    \subfloat[A QEC program without bugs.\label{fig:repetition_decoder_nobug}]{
        \centering
        \scalebox{.7}{\tikz{\node[fill=black!15,inner sep=2pt]{$\begin{aligned}
            & \qut{CNOT}{q_1,q_2}; \qmeasure{q_2}{s_1}; \qut{CNOT}{q_1,q_2}; \\
            & \qut{CNOT}{q_2,q_3}; \qmeasure{q_3}{s_2}; \qut{CNOT}{q_2,q_3}; \\
            & \cdots \\
            & \qut{CNOT}{q_{n-1},q_{n}}; \qmeasure{q_n}{s_{n-1}}; \qut{CNOT}{q_{n-1},q_{n}}; \\
            & \qut{CNOT}{q_n,q_1}; \qmeasure{q_1}{s_n}; \qut{CNOT}{q_n,q_1}; \\
            & \assign{\text{MWPM}(s_1,s_2,\ldots,s_n)}{r_1,r_2,\ldots,r_n}; \\
            & \qut{X[r_1==1]}{q_1}; \\
            & \qut{X[r_2==1]}{q_2}; \\
            & \cdots \\
            & \qut{X[r_n==1]}{q_n} \\
            & \vphantom{\qut{X[r_1*r_2*\cdots *r_{\lfloor \frac{n-1}{2}\rfloor}==1]}{q_1}}
        \end{aligned}$};}}
    }
    \hspace{1mm}
    \subfloat[A QEC program with a strange bug.\label{fig:repetition_decoder_bug}]{
        \centering
        \scalebox{.7}{\tikz{\node[fill=black!15,inner sep=2pt]{$\begin{aligned}
            & \qut{CNOT}{q_1,q_2}; \qmeasure{q_2}{s_1}; \qut{CNOT}{q_1,q_2}; \\
            & \qut{CNOT}{q_2,q_3}; \qmeasure{q_3}{s_2}; \qut{CNOT}{q_2,q_3}; \\
            & \cdots \\
            & \qut{CNOT}{q_{n-1},q_{n}}; \qmeasure{q_n}{s_{n-1}}; \qut{CNOT}{q_{n-1},q_{n}}; \\
            & \qut{CNOT}{q_n,q_1}; \qmeasure{q_1}{s_n}; \qut{CNOT}{q_n,q_1}; \\
            & \assign{\text{MWPM}(s_1,s_2,\ldots,s_n)}{r_1,r_2,\ldots,r_n}; \\
            & \qut{X[r_1==1]}{q_1}; \\
            & \qut{X[r_2==1]}{q_2}; \\
            & \cdots \\
            & \qut{X[r_n==1]}{q_n}; \\
            & \qut{X[r_1*r_2*\cdots *r_{\lfloor \frac{n-1}{2}\rfloor}==1]}{q_1} 
        \end{aligned}$};}}
    }
    \subfloat[Comments.]{
        \centering
        \scalebox{.7}{\tikz{\node[inner sep=2pt]{$\begin{aligned}
            & \vphantom{ \qut{CNOT}{q_1,q_2}; \qmeasure{q_2}{s_1}; \qut{CNOT}{q_1,q_2};} \text{// measure } Z_1Z_2\\
            & \vphantom{\qut{CNOT}{q_2,q_3}; \qmeasure{q_3}{s_2}; \qut{CNOT}{q_2,q_3};} \text{// measure } Z_2Z_3\\
            & \cdots \\
            & \vphantom{\qut{CNOT}{q_{n-1},q_{n}}; \qmeasure{q_n}{s_{n-1}}; \qut{CNOT}{q_{n-1},q_{n}};} \text{// measure } Z_{n-1}Z_n\\
            & \vphantom{\qut{CNOT}{q_n,q_1}; \qmeasure{q_1}{s_n}; \qut{CNOT}{q_n,q_1};} \text{// measure } Z_nZ_1\\
            & \vphantom{\assign{\text{MWPM}(s_1,s_2,\ldots,s_n)}{r_1,r_2,\ldots,r_n};} \text{// call MWPM} \\
            & \vphantom{\qut{X[r_1==1]}{q_1};} \text{// apply $X$ if $r_1=1$} \\
            & \vphantom{\qut{X[r_2==1]}{q_2};} \text{// apply $X$ if $r_2=1$}\\
            & \cdots \\
            & \vphantom{\qut{X[r_n==1]}{q_n};} \text{// apply $X$ if $r_n=1$} \\
            & \vphantom{\qut{X[r_1*r_2*\cdots *r_{\lfloor \frac{n-1}{2}\rfloor}==1]}{q_1}} \text{// a strange bug!}
        \end{aligned}$}}}
    }
    \caption{Two QEC programs for repetition code $C_n$. After measuring $Z$-checks $Z_1Z_2, \ldots, Z_{n-1}Z_n$, we additionally measure the observable $Z_nZ_1$ for the convenience of calling MWPM.}\label{fig:repetition_decoder}
\end{figure*}

In \cref{fig:repetition_decoder}, we provide two QEC programs for repetition code $C_n$, both written in the programming language defined in \S\ref{sec:qse:lang}.
The first one in \cref{fig:repetition_decoder_nobug} can correct any $X$ errors within $\lfloor \frac{n-1}{2}\rfloor$ locations of $n$ qubits, while the second one in \cref{fig:repetition_decoder_bug} has a strange bug that it will fail to correct errors when $X$ error occurred at qubits $q_1,q_2,\ldots,q_{\lfloor\frac{n-1}{2}\rfloor}$.

\subsection{Toric Codes}\label{sec:toric_code}
The Kitaev's toric codes~\cite{kitaev1997quantum, kitaev2003fault} are two-dimensional surface codes, a variant of which was implemented in the recent Google's QEC experiment~\cite{Acharya2022SuppressingQE}.
As shown in \cref{fig:toric_code}, a toric code $T_d$ with parameters $[[2d^2, 2, d]]$ is represented by a $d\times d$ two-dimensional lattice with period boundaries, hence a torus and so $T_d$ is called a toric code.
This lattice has $2d^2$ edges, representing $2d^2$ physical qubits of $T_d$.
Each face of the lattice defines a $Z$-check, observable $ZZZZ$, on the four edges surrounding the face.
Each vertex of the lattice defines an $X$-check, observable $XXXX$, on the four edges that touch the vertex.
There are $d^2$ $Z$-checks $\bar{A}_{1},\bar{A}_2,\ldots,\bar{A}_{d^2}$ (resp. $X$-checks $\bar{B}_1,\bar{B}_2,\ldots,\bar{B}_{d^2}$) and any $d^2-1$ of them are independent.
Besides, these $Z$-checks and $X$ checks are commuting with each other.

\begin{figure}[ht]

    \[
        \begin{array}[c]{c}
            \scalebox{.85}{\begin{tikzpicture}[scale=0.8]
    \draw[step=1cm,very thick,qubitcolor] (0,0) grid (5,5);
    \foreach \i in {0,1,2,3,4}{
        \path[dashed,draw=qubitcolor!20,very thick] (\i, 5) -- (\i+1,5);
        \path[dashed,draw=qubitcolor!20,very thick] (5, \i) -- (5,\i+1);
    }
    \foreach \i/\j/\k in {2/1/below,3/2/left,4/1/above,3/0/right}{
        \path[draw,xcheckcolor!80!black,ultra thick] (\i+1,\j+1) -- node[opacity=1,text=xcheckcolor!80!black,\k=-2pt,pos=0.4] {$X$} (4,2);
    }
    \node[circle, fill=xcheckcolor!90, inner sep=2pt] at (4,2) {};
    \node[rectangle, fill=zcheckcolor!70, inner sep=11pt,draw=zcheckcolor!80!black,ultra thick] at (2.5, 3.5) {};
    \foreach \i/\j/\k/\l/\m in {1/2/1/3/left, 1/2/2/2/below, 2/3/1/3/above, 2/3/2/2/right}{
        \path (\i+1,\j+1) -- node[opacity=1,text=zcheckcolor!80!black, \m=-2pt] {$Z$} (\k+1,\l+1);
    }
    \node[left] at (0,0) {$A$};
    \node[left] at (0,5) {$B$};
    \node[right,opacity=0.5] at (5,0) {$A'=A$};
    \node[right,opacity=0.5] at (5,5) {$B'=B$};
    \begin{scope}[overlay]
        \node[below] at (0,0) {$C$};
        \node[below] at (5,0) {$D$};
        \node[above,opacity=0.5,xshift=4mm] at (0,5) {$C'=C$};
        \node[above,opacity=0.5,xshift=-3mm] at (5,5) {$D=D'$};
    \end{scope}
    \foreach \i in {0,...,4}{
        \path[draw=green!50!black, line width=2.5pt, opacity=0.7] (1,\i) -- (1,\i+1);
        \path[draw=green!50!black, line width=2.5pt, opacity=0.7] (\i,1) -- (\i+1,1);
    }
\end{tikzpicture}}
        \end{array}
        \quad
        \begin{array}[c]{c}
            \scalebox{.85}{\begin{tikzpicture}[scale=0.8]
    \draw[step=1cm,very thick,qubitcolor!80] (0,0) grid (5,5);
    \foreach \i in {0,1,2,3,4}{
        \path[dashed,draw=qubitcolor!20,very thick] (\i, 5) -- (\i+1,5);
        \path[dashed,draw=qubitcolor!20,very thick] (5, \i) -- (5,\i+1);
    }
    \path[draw=blue,line width=2.5pt] (2,1) -- (4,1);
    \foreach \i/\j in {1/1,2/0,2/2}{
        \path[draw,xcheckcolor!80!black,ultra thick] (\i,\j) -- (2,1);
    }
    \foreach \i/\j in {4/0,5/1,4/2}{
        \path[draw,xcheckcolor!80!black,ultra thick] (\i,\j) -- (4,1);
    }
    \node[rectangle, fill=zcheckcolor!70, inner sep=11pt,draw=zcheckcolor!80!black,ultra thick] at (0.5, 3.5) {};
    \node[rectangle, fill=zcheckcolor!70, inner sep=11pt,draw=zcheckcolor!80!black,ultra thick] at (1.5, 3.5) {};
    \path[draw=red,line width=2.5pt] (1,3) -- (1,4);
    \node[circle, fill=xcheckcolor!70, inner sep=2pt] at (2, 1) {};
    \node[circle, fill=xcheckcolor!70, inner sep=2pt] at (4, 1) {};
\end{tikzpicture}}
        \end{array}
        \quad \setlength\arraycolsep{2pt}
            \def\arraystretch{1.2} \begin{array}[c]{cl}
            \mathtikz{\path[draw=qubitcolor,very thick] (0,0) -- (0.4,0);} & \text{Physical qubit} \\
            \mathtikz{\path[draw=red,line width=2.5pt] (0,0) -- (0.4,0);} & X\text{ error} \\
            \mathtikz{\path[draw=blue,line width=2.5pt] (0,0) -- (0.4,0);} & Z\text{ error} \\
            \mathtikz{
            \foreach \i/\j in {2.8/0,-2.8/0,0/2.8,0/-2.8}{
                \path[draw=xcheckcolor!80!black,ultra thick] (0,0) -- +(\i mm,\j mm);
            }
            \node[circle, fill=xcheckcolor!70, inner sep=3pt] {};} & X\text{-check} \\
            \mathtikz{
            \node[rectangle, fill=zcheckcolor!70, draw=zcheckcolor!80!black, ultra thick, inner sep=5pt] {};} & Z\text{-check} \\
            \mathtikz{\path[draw=green!50!black,line width=2.5pt] (0,0) -- (0.4,0);} & \text{Logical operator } \bar{Z}
        \end{array}
    \]
    \caption{\textbf{The toric code.} Physical qubits, represented by meat-brown edges, lie on a 2-D lattice with period boundaries ($AB=A'B', CD=C'D'$, hence a torus). Syndrome measurements consist of two types of checks: the $Z$-checks, represented by blue-gray faces, are observables $ZZZZ$ defined on the four edges that surround the faces; the $X$-checks, represented by red-gray vertices, are observables $XXXX$ defined on the four edges that touch the vertices.}\label{fig:toric_code}
\end{figure}
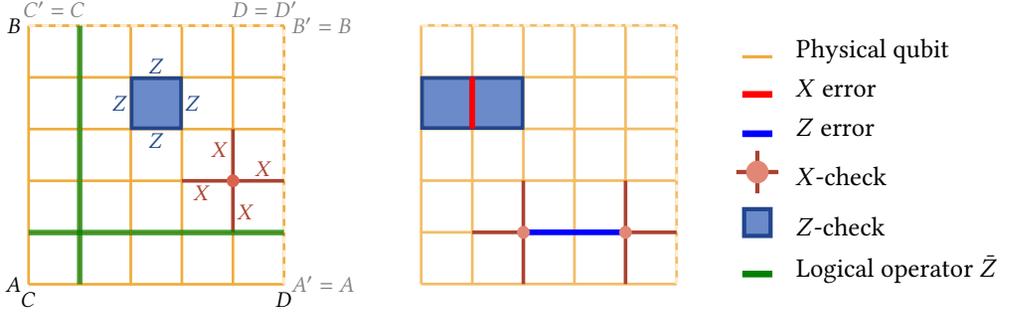

In addition to two types of these checks, there are two logical operators $\bar{Z}_1$ and $\bar{Z}_2$ defined by tensor products of Pauli $Z$ operators along the fundamental nontrivial cycles of the torus (the horizontal and vertical green lines in \Cref{fig:toric_code}).
The toric code $T_d$ encodes the logical computational basis state $\ket{x_1,x_2}_L$, $x_1,x_2=0,1$, of $2$ logical qubits as a $2d^2$-qubit stabilizer state of the stabilizer 
\[\langle (-1)^{x_1}\bar{Z}_1, (-1)^{x_2}\bar{Z}_2, A_1,A_2,\ldots, A_{d^2-1}, B_1, B_2,\ldots, B_{d^2-1}\rangle\]
The code space of $T_d = V_{\langle A_1,A_2,\ldots, A_{d^2-1}, B_1, B_2,\ldots, B_{d^2-1}\rangle}$.

The $X$ errors and $Z$ errors that occurred at physical qubits can be detected by $Z$-checks and $X$-checks, respectively.
For example, in \cref{fig:toric_code}, the red edge indicates an $X$ error occurred, then after measuring all $Z$-checks, we will find that the measurement outcomes of two $Z$-checks that share a common edge with this $X$ error are both $1$, while all other checks' outcomes are $0$.
With these measurement outcomes, we can detect this $X$ error.
In fact, decoding toric codes can be done by decoding X errors and Z errors, respectively. The pattern is similar to decoding repetition codes by MWPM algorithm that we used in \S\ref{sec:repetition_code}.

For toric code $T_d$, we adopt a QEC program with a bug similar to \cref{fig:repetition_decoder_bug} in evaluation.

\subsection{Quantum Tanner Codes}
There is an important class of QEC codes known as quantum low-density parity check (qLDPC) codes~\cite{breuckmann2021quantum}, where the low-density means that each check only acts on a constant number of physical qubits and each physical qubit only participates in a constant number of checks.
For example, the toric codes mentioned in \S\ref{sec:toric_code} are also qLDPC codes.
The low-density nature of qLDPC codes allows people to build efficient decoding scheme, while finding a family of good qLDPC codes with parameters $[[n, k = \Theta(n), d = \Theta(n)]]$ has plagued people for more than 20 years.
Excitingly, a major breakthrough made by \citet{panteleev2022asymptotically} gives the construction of asymptotically good qLDPC codes and a series of works of good qLDPC codes have come out in the past two years~\cite{leverrier2022quantum,dinur2022good,leverrier2023efficient,gu2022efficient}.

As a case study, we choose the quantum Tanner codes~\cite{leverrier2022quantum} that defines Tanner codes~\cite{tanner1981recursive} on the left-right Cayley complex, the cornerstone of good qLDPC codes~\cite{panteleev2022asymptotically} and c${}^3$ locally testable codes~\cite{dinur2022locally}, to obtain quantum codes.
However, a large number of physical qubits is needed to make the quantum Tanner codes be good qLDPC codes\footnote{\citet{dinur2022good} noticed an explicit construction of Ramanujan graphs\cite{lubotzky1988ramanujan} for the left-right Cayley complex, from which good qLDPC codes can be derived. Referring to the construction in~\cite{lubotzky1988ramanujan}, we need at least 47040 physical qubits for the quantum Tanner codes to be good qLDPC codes.}.
Since the construction of quantum Tanner codes is general, we follow the code construction, but use a small left-right Cayley complex and small Hamming codes, to construct quantum Tanner codes.


\subsubsection{Explicit Instance of Quantum Tanner Codes}
\label{sec:tanner_code}
In this section, we provide explicit instances of quantum Tanner codes used in our evaluation (\S\ref{sec:evaluation}), along with a brief description of their construction.
Interested readers can refer to the works of \citet{leverrier2022quantum,leverrier2023efficient} for more details.

\begin{definition}
    [Left-right Cayley complex (quadripartite version)~\cite{leverrier2022quantum}]\label{def:cayley}~
    
    \begin{minipage}[c]{0.75\linewidth}
        Given a group $G$ and two sets of generators $A = A^{-1}$ and $B = B^{-1}$,
        a left-right Cayley complex $X$ (see the right figure) is made up of vertices $V$, $A$-edges $E_A$, $B$-edges $E_B$ and squares $Q$, where
        \begin{itemize}
            \item $V = V_{00}\cup V_{10}\cup V_{10} \cup V_{11}$ with $V_{ij} = G\times \{ij\}$ (black points);
            \item $E_A = \{\{(g,i0), (ag, i1) \mid g\in G, a\in A, i\in\{0,1\}\}$ (red edges);
            \item $E_B = \{\{(g,0j), (gb, 1j) \mid g\in G, b\in B, j\in\{0,1\}\}$ (blue edges);
            \item $Q = \{\{(g,00),(ag,01),(gb,10),(agb,11)\}\mid g\in G, a\in A, b\in B\}$ \\ (squares surrounded by red and blue edges).
        \end{itemize}
    \end{minipage}
    \begin{minipage}[c]{0.23\linewidth}
        \[\mathtikz{
            \node[draw,ellipse,minimum width=0.4cm,minimum height=1.5cm] at (0,0.5){};
            \node[draw,ellipse,minimum width=0.4cm,minimum height=1.5cm] at (1.6,0.5) {};
            \node[draw,ellipse,minimum width=0.4cm,minimum height=1.5cm] at (0.6,1.8) {};
            \node[draw,ellipse,minimum width=0.4cm,minimum height=1.5cm] at (2.2,1.8) {};
            \node[circle,inner sep=1pt,draw,fill,label={below:\small $g$}] (a) at (0,0.6) {};
            \node[circle,inner sep=1pt,draw,fill,label={above:\small $ag$}] (b) at (0.6,1.8-0.2) {};
            \node[circle,inner sep=1pt,draw,fill,label={below:\small $gb$}] (c) at (1.6,0.7) {};
            \node[circle,inner sep=1pt,draw,fill,label={above:\small $agb$}] (d) at (2.2,1.8-0.4) {};
            \draw[thick,red] (a) -- (b) (c) -- (d);
            \draw[thick,blue] (a) -- (c) (b) -- (d);
            \node[overlay] at (-0.5,0) {$V_{00}$};
            \node[overlay] at (1.6+0.5,0) {$V_{10}$};
            \node[overlay] at (0.6-0.5,1.8) {$V_{01}$};
            \node[overlay] at (2.2+0.5,1.8) {$V_{11}$};
        }\]
        \vspace{5pt}
    \end{minipage}
\end{definition}

For simplicity, we assume that $A$ and $B$ are of the same cardinality $\Delta$ as in \cite{leverrier2022quantum}.
For any set $S$, we define $\FF_2^{S}$ as a linear space over the finite field $\FF_2=\ZZ/2\ZZ$ with basis $\{\vec{e}_s\mid s\in S\}$.
It is clear that $\FF_2^S \cong \FF_2^{\abs{S}}$.
Let $V_{0} = V_{00}\cup V_{11}, V_{1} = V_{01}\cup V_{10}$, we define two linear maps as follows:
\begin{itemize}
    \item $H_0:\FF_2^{Q}\to \FF_2^{A\times B}\otimes \FF_2^{V_0}$ with $H_0(\vec{e}_{\{(g,00),(ag,01),(gb,10),(agb,11)\}}) = \vec{e}_{(a, b)}\otimes \vec{e}_{(g,00)} + \vec{e}_{(a^{-1}, b^{-1})}\otimes \vec{e}_{(agb,11)}$. \\ $H_0$ represents the adjacent relationship between squares in $Q$ and vertices in $V_0$.
    \item $H_1:\FF_2^{Q}\to \FF_2^{A\times B}\otimes \FF_2^{V_1}$ with $H_1(\vec{e}_{\{(g,00),(ag,01),(gb,10),(agb,11)\}}) = \vec{e}_{(a^{-1}, b)}\otimes \vec{e}_{(ag,01)} + \vec{e}_{(a, b^{-1})}\otimes \vec{e}_{(gb,10)}$. \\ $H_1$ represents the adjacent relationship between squares in $Q$ and vertices in $V_1$.
\end{itemize}
To make $H_0$ and $H_1$ have orthogonal properties on $A\times B$, we need two (classical) linear codes $C_A$ and $C_B$ of length $\Delta$ with their dual codes $C_A^{\bot}$ and $C_B^{\bot}$.
Let 
\begin{itemize}
    \item $H_A:\FF_2^\Delta\to\FF_2^{m_A}$ be the check matrix of $C_A$ with $m_A = \dim(C_A)$.
    \item $H_A^{\bot}:\FF_2^\Delta\to\FF_2^{\Delta-m_A}$ be the check matrix of $C_A^{\bot}$.
    \item $H_B:\FF_2^\Delta\to\FF_2^{m_B}$ be the check matrix of $C_B$ with $m_B = \dim(C_B)$.
    \item $H_B^{\bot}:\FF_2^\Delta\to\FF_2^{\Delta-m_B}$ be the check matrix of $C_B^{\bot}$.
\end{itemize}

Now, consider the following sequences of maps:
\[
    \begin{array}{cccccccl}
        \FF_2^{\Delta^2 \abs{G}} &\overto{H_Q}& \FF_2^{Q} &\overto{H_0}& \FF_2^{A\times B}\otimes \FF_2^{V_0} &\overto{H_{A\times B}\otimes H_{V_0}}& \FF_2^{\Delta^2}\otimes \FF_2^{2\Delta\abs{G}} &\overto{(H_A\otimes H_B)\otimes \Id} \FF_2^{2m_Am_B\Delta\abs{G}} \\
        \FF_2^{\Delta^2 \abs{G}} &\overto{H_Q}& \FF_2^{Q} &\overto{H_1}& \FF_2^{A\times B}\otimes \FF_2^{V_1} &\overto{H_{A\times B}\otimes H_{V_1}}& \FF_2^{\Delta^2}\otimes \FF_2^{2\Delta\abs{G}} & \overto{(H_A^{\bot}\otimes H_B^{\bot})\otimes \Id} \FF_2^{2(\Delta-m_A)(\Delta-m_B)\Delta\abs{G}}
    \end{array}
\]
where
\begin{itemize}
    \item $H_Q$ is an isomorphism map for $\FF_2^{\Delta^2 \abs{G}} \cong \FF_2^{Q}$;
    \item $H_{A\times B}$ is an isomorphism map for $\FF_2^{A\times B} \cong \FF_2^{\Delta^2}$;
    \item $H_{V_0}$ is an isomorphism map for $\FF_2^{V_0} \cong \FF_2^{2\Delta \abs{G}}$;
    \item $H_{V_1}$ is an isomorphism map for $\FF_2^{V_1} \cong \FF_2^{2\Delta \abs{G}}$;
\end{itemize}
we define
\begin{align*}
    H_X &\triangleq ((H_A\otimes H_B)\otimes \Id)\circ (H_{A\times B}\otimes H_{V_0}) \circ H_0 \circ H_Q &&: \FF_2^{\Delta^2\abs{G}} \to \FF_2^{2m_Am_B\Delta\abs{G}}, \\
    H_Z &\triangleq ((H_A^{\bot}\otimes H_B^{\bot})\otimes \Id)\circ (H_{A\times B}\otimes H_{V_1}) \circ H_1 \circ H_Q &&: \FF_2^{\Delta^2\abs{G}} \to \FF_2^{2(\Delta-m_A)(\Delta-m_B)\Delta\abs{G}}.
\end{align*}
Without causing ambiguity, we also use \[H_X \in \FF_2^{2m_Am_B\Delta\abs{G} \times \Delta^2\abs{G}}, H_Z \in \FF^{2(\Delta-m_A)(\Delta-m_B)\Delta\abs{G} \times \Delta^2\abs{G}}\]
for their matrix representations.
According to~\cite{leverrier2022quantum}, we can prove that $H_XH_Z^{\intercal} = 0$, therefore $(H_X, H_Z)$ defines a CSS code~\cite{calderbank1996good,steane1996error}, which is a specific instance of stabilizer code. In particular:
\begin{itemize}
    \item The column indexes of $H_X,H_Z$ represent a total of $\Delta^2\abs{G}$ physical qubits.
    \item Each row $\vec{x} = (x_1,x_2,\ldots,x_{\Delta^2\abs{G}}) \in \FF_2^{\Delta^2\abs{G}} $ of $H_X$ defines an $X$-check 
    \[P_X(\vec{x}) = X^{x_1}\otimes X^{x_2}\otimes \cdots \otimes X^{x_{\Delta^2\abs{G}}},\]
    where $X^0$ is Pauli $I$, $X^1$ is Pauli $X$.
    \item Each row $\vec{z} = (z_1,z_2,\ldots,z_{\Delta^2\abs{G}}) \in \FF_2^{\Delta^2\abs{G}} $ of $H_Z$ defines a $Z$-check
    \[P_Z(\vec{z}) = Z^{z_1}\otimes Z^{z_2}\otimes \cdots \otimes Z^{z_{\Delta^2\abs{G}}},\]
    where $Z^0$ is Pauli $I$, $Z^1$ is Pauli $Z$.
    \item The quantum Tanner code $C_{G,A,B,C_A,C_B}$ is the space of states stabilized by these checks.
\end{itemize}
Using these checks, we provide a QEC program for the quantum Tanner code $C_{G,A,B,C_A,C_B}$ in pseudocode as follows.
\begin{algorithm}[h]
	\caption{Decoding quantum Tanner code $C_{G,A,B,C_A,C_B}$}
	\label{alg:path_filter}
	\begin{algorithmic}[1]
        \State{$\vec{S}_X \gets$ An empty $2m_Am_B\Delta\abs{G} \times 1$ vector}
        \State{$\vec{S}_Z \gets$ An empty $2(\Delta-m_A)(\Delta-m_B)\Delta\abs{G} \times 1$ vector}
        \For{$j \gets 1,\cdots, 2m_Am_B\Delta\abs{G}$} \Comment{Syndrome measurements of $X$-checks}
            \State{$\vec{x} \gets \mathrm{Row}(H_X, j)$}
            \State{$\vec{S}_X[j] \gets $ measure $X$-check $P_X(\vec{x})$}
        \EndFor
        \For{$j \gets 1,\cdots, 2(\Delta-m_A)(\Delta-m_B)\Delta\abs{G}$} \Comment{Syndrome measurements of $Z$-checks}
            \State{$\vec{z} \gets \mathrm{Row}(H_Z, j)$}
            \State{$\vec{S}_Z[j] \gets $ measure $Z$-check $P_Z(\vec{z})$}
        \EndFor
        \State{$\vec{R}_X, \vec{R}_Z \gets \textsc{DecodeSyndromes}(\vec{S}_X, \vec{S}_Z)$} \Comment{Call Algorithm 1 in~\cite{leverrier2023efficient}}
        \For{$j \gets 1,\cdots, \Delta^2\abs{G}$}
            \If{$\vec{R}_X[j] = 1$}
                \State{Perform $X$ gate at the $j$-th qubit}
            \EndIf
            \If{$\vec{R}_Z[j] = 1$}
                \State{Perform $Z$ gate at the $j$-th qubit}
            \EndIf
        \EndFor
	\end{algorithmic}
\end{algorithm}

For external call of \textsc{DecodeSyndromes}, the condition $C_{\textsc{DecodeSyndromes}}(\vec{S}_X, \vec{S}_Z, \vec{R}_X, \vec{R}_Z)$ is defined as
\[ (\vec{S}_X == H_X \vec{R}_Z)\land(\lvert\vec{R}_Z\rvert_{\text{Hamming}} \leq n_Z) \land (\vec{S}_Z == H_Z\vec{R}_X)\land(\lvert\vec{R}_X\rvert_{\text{Hamming}} \leq n_X), \]
where $n_X$ is the number of injected $X$ errors and $n_Z$ is the number of injected $Z$ errors.

\begin{example}
    [Explicit Instances used in evaluation (\S\ref{sec:evaluation})]
    Let $m,k\in\ZZ$, the $G, A, B, C_A$ and $ C_B$ required by the construction are chosen as follows: 
    \begin{itemize}
        \item $G = (\ZZ_{k7^m}, +)$.
        \item $A = B = k7^{m-1} \cdot \{0,1,2,3,4,5,6\}$ with $\Delta = 7$.
        \item $C_A$ is the $[7,4,3]$ Hamming code.
        \item $C_B$ is the $[7,3,3]$ dual code of the$[7,4,3]$ Hamming code.
    \end{itemize}
    For our evaluation, we have chosen the values of $(m,k) = (1,1), (1,2), (1,3)$, and $(1,4)$. These values result in a total number of physical qubits corresponding to each code of $\Delta^2\abs{G} = k7^{m+2} = 343, 686, 1029$, and $1372$, respectively.
\end{example}

\section{Prototype Implementation}
\label{sec:implementation}
We implemented our  general QSE framework together with the special support of symbolic stabilizer states in a prototype tool called \QuantumSE{} as a Julia~\cite{bezanson2017julia} package.

\subsubsection*{Quantum DSL}
To conveniently write QEC programs, e.g., the two programs in \cref{fig:repetition_decoder}, in Julia, we implement a quantum DSL with more features, including for-loops, list comprehension, and subroutine call, beyond the quantum programs defined in \S\ref{sec:qse:lang}.
Thanks to Julia's meta-programming support, \QuantumSE{} can handle these features by Julia's native statements.
So, only the QSE rules in \cref{fig:qse_rules} remain to consider for this quantum DSL.

\subsubsection*{Symbolic stabilizer states}
Currently, \QuantumSE{} only supports using symbolic stabilizer states as symbolic quantum states, but because of the multiple dispatch of Julia, we can easily extend it to other symbolic quantum states in the future.

We implemented the unitaries and measurements over symbolic stabilizer states in \QuantumSE{} based on the high-performance simulator, \texttt{chp.c}, of stabilizer circuits by \citet{aaronson2004improved}.
The basic idea of simulating stabilizer circuits is using a tableau to track the dynamic of stabilizer states' generators.
The tableau consists of binary variables $x_{jk}, z_{jk}$ with $x_{jk}z_{jk}=00,10,01,11$ means that the $j$-th generator contains $I, X, Z, Y$, respectively, at the $k$-th qubit; and variables $r_j=0,1,2,3$ means that the phase of the $j$-th generator is $(i)^{r_j}$.
For $n$-qubits stabilizer circuits, this idea offers us a complexity of $\mathcal{O}(n)$ to perform single-qubit Clifford gates and CNOT gate, and a complexity $\mathcal{O}(n^3)$ to perform computational basis measurements~\cite{gottesman1998heisenberg}.
\citet{aaronson2004improved} further reduced the complexity of measurements to $\mathcal{O}(n^2)$ by an improved tableau algorithm with destabilizers and provided its corresponding implementation \texttt{chp.c}.
By symbolizing the phase part $r_i$ of the improved tableau algorithm (\texttt{chp.c}), our \QuantumSE{} then supports symbolic stabilizer states.
We should point out that here the (symbolic) phases do not affect the control flow of the improved tableau algorithm, thus \QuantumSE{} can handle the symbolic stabilizer states with the same complexity of \texttt{chp.c}.
This is also an important reason why we define such symbolic stabilizer states in \cref{def:sss}.

Since \texttt{chp.c} is a C program, one may wonder why don't we use the classic SE tools, e.g., the KLEE~\cite{cadar2008klee}, to run \texttt{chp.c}?
This can certainly be done, but the main issue is that the efficiency of directly using the classic SE tool is not high enough, even if applying the technique of \emph{concolic execution}~\cite{sen2005cute} that executes the program with concrete variables.
We have tried to use KLEE to run \texttt{chp.c} with only concrete variables, but the efficiency is not good.
This is because KLEE interprets programs rather than running them directly on native machines.
Interestingly, another approach of classical SE based on the compilation, e.g., the SymCC~\cite{poeplau2020symbolic}, could provide a more efficient SE than the interpreter approach of KLEE.
But we didn't adopt it because it required a lot of engineering effort to integrate it into our framework.

In addition, there are some more efficient implementations for Clifford circuits than \texttt{chp.c}, e.g., QuantumClifford.jl~\cite{krastanov} and Google's Stim~\cite{gidney2021stim} based on tableau algorithm with some theoretical and engineering optimizations;
and GraphSim~\cite{anders2006fast} based on graph-state representation which may provide complexity of $\mathcal{O}(n\log{n})$ for measurements in some special cases.
However, as we see in evaluation (\S\ref{sec:rq3}), the efficiency of the improved tableau algorithm and \texttt{chp.c} is good enough, and the bottleneck of \QuantumSE{}'s efficiency depends on the SMT solver.

\subsubsection*{External call}
The external call of a classical function $F$ in QSE rules (\cref{fig:qse_rules}) requires the user to explicitly give the condition satisfied by the input and output of the function $F$.
With \QuantumSE{}, the user needs to write down the conditions wrapped by a function in Julia.
Then, during quantum SE, \QuantumSE{} will call this function to obtain the conditions.

\subsubsection*{SMT solvers}
\QuantumSE{} uses bit-vector theory from SMT-LIB~\cite{BarFT-SMTLIB} to model symbolic variables related to symbolic stabilizer states.
During quantum SE, \QuantumSE{} uses a wrapper, Z3.jl, of the Z3 SMT solver~\cite{moura2008z3} to manipulate symbolic expressions.
Finally, when solving the assertions, \QuantumSE{} calls the Bitwuzla SMT solver~\cite{DBLP:conf/cav/NiemetzP23}\footnote{We choose Bitwuzla as it is a winner of the competition in the quantifier-free bit-vector logic category in SMT-COMP 2022 (see \url{https://smt-comp.github.io/2022/results/qf-bitvec-single-query}).}.

\section{Omitted Proofs in Section~\ref{sec:soundness}}\label{app:soundess_prf}~

\correctnessut*
\begin{proof} Notice the following facts:
    \begin{itemize}
        \item With operational semantics in \cref{fig:operational} and $\cfg{\qut{U}{\bar{q}}}{\sigma}{\rho} \to \cfg{\terminate}{\sigma'}{\rho'}$, we have
        $\sigma' = \sigma, \rho' = U_{\bar{q}}\rho U_{\bar{q}}^{\dagger}$;
        \item With the definition of $ut$ in \cref{def:operations}, we have $ut(U, \bar{q}, \srho) = U_{\bar{q}}\srho U_{\bar{q}}^{\dagger}$;
        \item With $(\sigma, \rho) \models (\ssigma, \srho)$ and \cref{def:satisfaction}, there is a valuation $V$ such that $\sigma = V(\ssigma), \rho = V(\srho)$ and $V(\varphi) = \texttt{true}$.
    \end{itemize}
    Then we have \begin{align*}
        V(\ssigma) &= \sigma = \sigma',\\
        V(ut(U, \bar{q}, \srho)) &= V(U_{\bar{q}}\srho U_{\bar{q}}^{\dagger}) = U_{\bar{q}}V(\srho) U_{\bar{q}}^{\dagger} = U_{\bar{q}}\rho U_{\bar{q}}^{\dagger} = \rho', \\
        V(\varphi) &= \texttt{true}.
    \end{align*}
    Therefore, by \cref{def:satisfaction}, we have \[(\sigma', \rho') \models_{\varphi} (\ssigma, ut(U, \bar{q}, \srho)).\]

\end{proof}

\correctnessm*
\begin{proof}
    Notice the following facts:
    \begin{itemize}
        \item With the definition of $m$ in \cref{def:operations}, we have \(\srho'(s) = \frac{\ketbra[q]{s}{s}\srho\ketbra[q]{s}{s}}{\tr(\ketbra[q]{s}{s}\srho)};\)
        \item With $(\sigma, \rho) \models (\ssigma, \srho)$ and \cref{def:satisfaction}, there is a valuation $V$ such that $\sigma = V(\ssigma), \rho = V(\srho)$ and $V(\varphi) = \<true>$.
    \end{itemize}
    The goal is divided into the two cases of the operational semantics of measurement statement in \cref{fig:operational}:
    \begin{enumerate}
        \item $\cfg{\qmeasure{q}{c}}{\sigma}{\rho} \overto{\tr(\ketbra[q]{0}{0}\rho)} \cfg{\terminate}{\sigma[0/c]}{\frac{\ketbra[q]{0}{0}\rho\ketbra[q]{0}{0}}{\tr(\ketbra[q]{0}{0}\rho)}}$.
        In this case, we have \[\sigma' = \sigma[0/c], \quad \rho' = \frac{\ketbra[q]{0}{0}\rho\ketbra[q]{0}{0}}{\tr(\ketbra[q]{0}{0}\rho)}.\]
        Let $V' = V[s\mapsto 0]$, we have $V'(\ssigma) = V(\ssigma) = \sigma, V'(\srho) = V(\srho) = \rho, V'(\varphi) = V(\varphi) = \<true>$ since $s$ is a newly introduced symbol that does not appear in $\ssigma,\srho$ and $\varphi$, then
        \begin{align*}
            V'(\ssigma[s/c]) &= V'(\ssigma)[0/c] = \sigma[0/c] = \sigma', \\
            V'(\srho'(s)) &= V'\left(\frac{\ketbra[q]{s}{s}\srho\ketbra[q]{s}{s}}{\tr(\ketbra[q]{s}{s}\srho)}\right) = \frac{\ketbra[q]{0}{0}V'(\srho)\ketbra[q]{0}{0}}{\tr(\ketbra[q]{0}{0}V'(\srho))} = \frac{\ketbra[q]{0}{0}\rho\ketbra[q]{0}{0}}{\tr(\ketbra[q]{0}{0}\rho)} = \rho',\\
            V'(\varphi) &= \<true>.
        \end{align*}
        Therefore, by \cref{def:satisfaction}, we have \[(\sigma', \rho') \models_{\varphi} (\ssigma[s/c], \srho'(s)).\]
        \item $\cfg{\qmeasure{q}{c}}{\sigma}{\rho} \overto{\tr(\ketbra[q]{1}{1}\rho)} \cfg{\terminate}{\sigma[1/c]}{\frac{\ketbra[q]{1}{1}\rho\ketbra[q]{1}{1}}{\tr(\ketbra[q]{1}{1}\rho)}}$.Similar to (1), we only need to take $V' = V[s\mapsto 1]$.
    \end{enumerate}
\end{proof}

\soundness*
\begin{proof}
    With $(\sigma, \rho)\models_{\varphi} (\ssigma, \srho)$ and \cref{def:satisfaction}, there is a valuation $V$ such that $\sigma = V(\ssigma), \rho = V(\srho), V(\varphi) = \texttt{true}$.
    We then prove the theorem by induction through the program structure.
    \begin{itemize}
        \item $S \equiv \assign{e}{x}$. We have \[\cfg{\assign{e}{x}}{\sigma}{\rho} \to \cfg{\terminate}{\sigma[\sigma(e)/x]}{\rho} \] and \[ \scfg{\assign{e}{x}}{\ssigma}{\srho}{P}{\varphi} \to \scfg{\terminate}{\ssigma[\ssigma(e)/x]}{\srho}{P}{\varphi},\]
        which imply \[ \sigma' = \sigma[\sigma(e)/x],\quad \rho' = \rho \] and \[ \ssigma' = \ssigma[\ssigma(e)/x],\quad \srho' = \srho,\quad \varphi' = \varphi.\]
        Then,
        \begin{align*}
            V(\ssigma') &= V(\ssigma[\ssigma(e)/x]) = V(\ssigma)[V(\ssigma(e))/x] = V(\ssigma)[V(\ssigma)(e)/x] = \sigma[\sigma(e)/x] = \sigma', \\
            V(\srho') &= V(\srho) = \rho = \rho', \\
            V(\varphi') &= V(\varphi) = \texttt{true}.
        \end{align*}
        Therefore, with \cref{def:satisfaction}, we have $(\sigma',\rho') \models_{\varphi'} (\ssigma', \srho')$.
        \item $S\equiv \assign{F(\bar{x})}{\bar{y}}$. We have
        \[\cfg{\assign{F(\bar{x})}{\bar{y}}}{\sigma}{\rho} \to \cfg{\terminate}{\sigma[F(\sigma(\bar{x}))/\bar{y}]}{\rho} \] and \[ \scfg{\assign{F(\bar{x})}{\bar{y}}}{\ssigma}{\srho}{P}{\varphi} \to \scfg{\terminate}{\ssigma[\bar{s}_{\bar{y}}/\bar{y}]}{\srho}{P}{\varphi \land C_{F}(\ssigma(\bar{x}), \bar{s}_{\bar{y}})},\]
        which imply \[ \sigma' = \sigma[F(\sigma(\bar{x}))/\bar{y}],\quad \rho' = \rho \] and \[ \ssigma' = \ssigma[\bar{s}_{\bar{y}}/\bar{y}],\quad \srho' = \srho,\quad \varphi' = \varphi \land C_{F}(\ssigma(\bar{x}), \bar{s}_{\bar{y}}).\]
        Let $V' = V[\bar{s}_{\bar{y}}\mapsto F(\sigma(\bar{x}))]$, we have $V'(\ssigma) = V(\ssigma) = \sigma, V'(\srho) = V(\srho) = \rho, V'(\varphi) = V(\varphi) = \texttt{true}$ since $\bar{s}_{\bar{y}}$ are newly introduced symbols that do not appear in $\ssigma, \srho$ and $\varphi$, then
        \begin{align*}
            V'(\ssigma') &= V'(\ssigma[\bar{s}_{\bar{y}}/\bar{y}]) = V'(\ssigma)[F(\sigma(\bar{x}))/\bar{y}] = \sigma[F(\sigma(\bar{x}))/\bar{y}] = \sigma' \\
            V'(\srho') &= \rho = \rho' \\
            V'(\varphi') &= V'(\varphi \land C_{F}(\ssigma(\bar{x}), \bar{s}_{\bar{y}})) = V'(\varphi) \land V'(C_{F}(\ssigma(\bar{x}), \bar{s}_{\bar{y}})) \\
            &= \texttt{true} \land C_F(V'(\ssigma)(\bar{x}),F(\sigma(\bar{x}))) \\
            &= \texttt{true} \land C_F(\sigma(\bar{x}),F(\sigma(\bar{x}))) \\
            & = \texttt{true} \land \texttt{true} = \texttt{true} \tag{see caption of \cref{fig:qse_rules}}
        \end{align*}
        Therefore, with \cref{def:satisfaction}, we have $(\sigma',\rho') \models_{\varphi'} (\ssigma', \srho')$.
        \item $S \equiv \qut{U}{\bar{q}}$. We have
        \[\scfg{\qut{U}{\bar{q}}}{\ssigma}{\srho}{P}{\varphi} \to \scfg{\terminate}{\ssigma}{ut(U,\bar{q},\srho)}{P}{\varphi}\]
        which imply
        \[ \ssigma' = \ssigma,\quad \srho' = ut(U,\bar{q},\srho),\quad \varphi' = \varphi.\]
        With $(\sigma, \rho)\models_{\varphi} (\ssigma, \srho)$ and \cref{lem:ut}, we can see $(\sigma', \rho')\models_{\varphi} (\ssigma, ut(U,\bar{q},\srho))$, then
        \[(\sigma', \rho')\models_{\varphi'} (\ssigma', \srho').\]
        \item $S \equiv \qmeasure{q}{c}$. We have
        \[\scfg{\qmeasure{q}{c}}{\ssigma}{\srho}{P}{\varphi} \to \scfg{\terminate}{\ssigma[s/c]}{\srho'(s)}{P\cup\{(s,p(s))\}}{\varphi}, \]
        where $(s,p(s),\srho'(s)) = m(q,\srho)$. Then 
        \[\ssigma' = \ssigma[s/c], \quad \srho' = \srho'(s), \quad \varphi' = \varphi\]
        With $(\sigma, \rho)\models_{\varphi} (\ssigma, \srho)$ and \cref{lem:m}, we can see $(\sigma', \rho')\models_{\varphi} (\ssigma[s/c], \srho'(s))$, then
        \[(\sigma', \rho')\models_{\varphi'} (\ssigma', \srho').\]
        \item $S\equiv S_1;S_2$. We have 
        \[\inference{\cfg{S_1}{\sigma}{\rho} \overto{p} \cfg{S_1'}{\sigma'}{\rho'}}{\cfg{S_1;S_2}{\sigma}{\rho} \overto{p} \cfg{S_1';S_2}{\sigma'}{\rho'}}\]
        and \[\inference{\scfg{S_1}{\ssigma}{\srho}{P}{\varphi} \to \scfg{S_1'}{\ssigma'}{\srho'}{P'}{\varphi'}}{\scfg{S_1;S_2}{\ssigma}{\srho}{P}{\varphi} \to \scfg{S_1';S_2}{\ssigma'}{\srho'}{P'}{\varphi'}}\]
        Suppose $\cfg{S_1}{\sigma}{\rho} \overto{p} \cfg{S_1'}{\sigma'}{\rho'}$ and $\scfg{S_1}{\ssigma}{\srho}{P}{\varphi} \to \scfg{S_1'}{\ssigma'}{\srho'}{P'}{\varphi'}$.
        By the inductive hypothesis, we can see \((\sigma', \rho')\models_{\varphi'} (\ssigma', \srho')\), which is the exact goal we need to prove for $S_1;S_2$.
        \item $S\equiv \qqif{b}{S_1}{S_2}$. Consider the following two cases:
        \begin{enumerate}
            \item $\sigma\models b$. In this case, we have \[ \cfg{\qqif{b}{S_1}{S_2}}{\sigma}{\rho} \to \cfg{S_1}{\sigma}{\rho}\]
            and \[ \scfg{\qqif{b}{S_1}{S_2}}{\ssigma}{\srho}{P}{\varphi} \to \scfg{S_1}{\ssigma}{\srho}{P}{\varphi \land \ssigma(b)},\]
            which imply
            \[\sigma' = \sigma, \quad \rho' = \rho \]
            and \[ \ssigma' = \ssigma, \quad \srho' = \srho,\quad \varphi' = \varphi \land \ssigma(b).\]
            Then,
            \begin{align*}
                V(\ssigma') &= V(\ssigma) = \sigma = \sigma', \\
                V(\srho') &= V(\srho) = \rho = \rho', \\
                V(\varphi') &= V(\varphi \land \ssigma(b)) = V(\varphi) \land V(\ssigma(b)) = \<true> \land V(\ssigma)(b) = \sigma(b) = \<true>.
            \end{align*}
            Therefore, with \cref{def:satisfaction}, we have $(\sigma', \rho') \models_{\varphi'} (\ssigma', \srho')$.
            \item $\sigma \models \neg b$. In this case, we have \[ \cfg{\qqif{b}{S_1}{S_2}}{\sigma}{\rho} \to \cfg{S_2}{\sigma}{\rho}\]
            and \[ \scfg{\qqif{b}{S_1}{S_2}}{\ssigma}{\srho}{P}{\varphi} \to \scfg{S_2}{\ssigma}{\srho}{P}{\varphi \land \neg\ssigma(b)},\]
            which imply
            \[\sigma' = \sigma, \quad \rho' = \rho \]
            and \[ \ssigma' = \ssigma, \quad \srho' = \srho,\quad \varphi' = \varphi \land \neg\ssigma(b).\]
            Then,
            \begin{align*}
                V(\ssigma') &= V(\ssigma) = \sigma = \sigma', \\
                V(\srho') &= V(\srho) = \rho = \rho', \\
                V(\varphi') &= V(\varphi \land \neg\ssigma(b)) = V(\varphi) \land V(\ssigma(\neg b)) = \<true> \land V(\ssigma)(\neg b) =\sigma(\neg b) = \<true>.
            \end{align*}
            Therefore, with \cref{def:satisfaction}, we have $(\sigma', \rho') \models_{\varphi'} (\ssigma', \srho')$.
        \end{enumerate}
    \end{itemize}
    For the additional condition $P' = P\cup\{(s, p(s))\}$, we only need consider $S \equiv \qmeasure{q}{c}$. We have
    \[\scfg{\qmeasure{q}{c}}{\ssigma}{\srho}{P}{\varphi} \to \scfg{\terminate}{\ssigma[s/c]}{\srho'(s)}{P\cup\{(s,p(s))\}}{\varphi}, \]
    where $(s,p(s),\srho'(s)) = m(q,\srho) = \left(s, \tr\bigl(\ketbra[q]{s}{s}\srho\bigr), \frac{\ketbra[q]{s}{s}\srho\ketbra[q]{s}{s}}{\tr(\ketbra[q]{s}{s}\srho)}\right)$, then \[\ssigma' = \ssigma[s/c], \quad \srho' = \frac{\ketbra[q]{s}{s}\srho\ketbra[q]{s}{s}}{\tr(\ketbra[q]{s}{s}\srho)}, \quad \varphi' = \varphi, \quad p(s) = \tr\bigl(\ketbra[q]{s}{s}\srho\bigr).\]
    We divide the discussion into two cases like \cref{lem:m}:
    \begin{enumerate}
        \item $\cfg{\qmeasure{q}{c}}{\sigma}{\rho} \overto{\tr(\ketbra[q]{0}{0}\rho)} \cfg{\terminate}{\sigma[0/c]}{\frac{\ketbra[q]{0}{0}\rho\ketbra[q]{0}{0}}{\tr(\ketbra[q]{0}{0}\rho)}}$. In this case, we have 
        \[\sigma' = \sigma[0/c],\quad \rho' = \frac{\ketbra[q]{0}{0}\rho\ketbra[q]{0}{0}}{\tr(\ketbra[q]{0}{0}\rho)}, \quad p = \tr(\ketbra[q]{0}{0}\rho).\]
        Let $V' = V[s\mapsto 0]$, we have $V'(\ssigma) = V(\ssigma) = \sigma, V'(\srho) = V(\srho) = \rho, V'(\varphi) = V(\varphi) = \<true>$ since $s$ is a newly introduced symbol that does not appear in $\ssigma, \srho$ and $\varphi$, then
        \begin{align*}
            V'(\sigma') &= V'(\ssigma[s/c]) = V'(\ssigma)[0/c] = \sigma[0/c] = \sigma', \\
            V'(\srho') &=  V'\left(\frac{\ketbra[q]{s}{s}\srho\ketbra[q]{s}{s}}{\tr(\ketbra[q]{s}{s}\srho)}\right) = \frac{\ketbra[q]{0}{0}V'(\srho)\ketbra[q]{0}{0}}{\tr(\ketbra[q]{0}{0}V'(\srho))} = \frac{\ketbra[q]{0}{0}\rho\ketbra[q]{0}{0}}{\tr(\ketbra[q]{0}{0}\rho)} = \rho',\\
            V'(\varphi') &= V'(\varphi) = \<true>,\\
            V'(p(s)) &= V'\left(\tr\bigl(\ketbra[q]{s}{s}\srho\bigr)\right) = \tr\bigl(\ketbra[q]{0}{0}V'(\srho)\bigr) = \tr\bigl(\ketbra[q]{0}{0}\rho\bigr) = p.
        \end{align*}
        \item $\cfg{\qmeasure{q}{c}}{\sigma}{\rho} \overto{\tr(\ketbra[q]{1}{1}\rho)} \cfg{\terminate}{\sigma[1/c]}{\frac{\ketbra[q]{1}{1}\rho\ketbra[q]{1}{1}}{\tr(\ketbra[q]{1}{1}\rho)}}$. The same as (1), we only need to take $V' = V[s\mapsto 1]$.
    \end{enumerate}
\end{proof}

\section{Omitted Proofs in Section~\ref{sec:symbolic_stabilizer}}\label{app:adeqacy_prf}

\subsection{Correctness for \cref{def:qo_sss}}
We explain the correctness of \cref{def:qo_sss} according to the stabilizer formalism~\cite[Chapter 10.5]{nielsen2010quantum}.
Given a symbolic stabilizer state $$\srho = \langle (-1)^{f_1(s_1,\ldots,s_m)}P_1,\ldots,(-1)^{f_n(s_1,\ldots,s_m)}P_n\rangle$$ Then: 
\begin{itemize}
    \item The unitary transformation by a Clifford gate $V$ on  qubits $\bar{q}$ changes the stabilizer to
    \[\langle (-1)^{f_1(s_1,\ldots,s_m)}V_{\bar{q}}P_1V_{\bar{q}}^{\dagger},\ldots,(-1)^{f_n(s_1,\ldots,s_m)}V_{\bar{q}}P_nV_{\bar{q}}^{\dagger}\rangle\]
    as mentioned in preliminary (\S\ref{sec:qec&stabilizer}). Thus, the unitary transformation function defined for $\srho$ is
    \[ut(V, \bar{q}, \srho) = \langle (-1)^{f_1(s_1,\ldots,s_m)}V_{\bar{q}}P_1V_{\bar{q}}^{\dagger},\ldots,(-1)^{f_n(s_1,\ldots,s_m)}V_{\bar{q}}P_nV_{\bar{q}}^{\dagger}\rangle\]
    \item The computational basis measurement on qubit $q$, which corresponds  to the observable $Z_q$, is defined in the following two cases:
    \begin{enumerate}[leftmargin=6mm]
        \item $Z_q$ commutes with all $P_j$. In this case, there exist a Boolean value $b$ and a list of indexes $1\leq j_1,j_2,\ldots,j_k\leq n$ with $1\leq k \leq n$ such that $(-1)^bZ_q = P_{j_1}P_{j_2}\cdots P_{j_k}$~\cite[Chapter 10.5.3]{nielsen2010quantum}. Then, we have that for any $b_1,b_2,\ldots, b_m$,
        \[ (-1)^{b\oplus f_{j_1}(b_1,\ldots,b_m)\oplus\cdots\oplus f_{j_k}(b_1,\ldots,b_m)}Z_q = \bigl((-1)^{f_{j_1}(b_1,\ldots,b_m)}P_{j_1}\bigr)\cdots\bigl((-1)^{f_{j_k}(b_1,\ldots,b_m)}P_{j_k}\bigr), \]
        which means $(-1)^{b\oplus f_{j_1}(b_1,\ldots,b_m)\oplus\cdots\oplus f_{j_k}(b_1,\ldots,b_m)}Z_q$ belongs to the stabilizer \[\cS_{b_1,b_2,\ldots,b_m} = \langle (-1)^{f_1(b_1,\ldots,b_m)}P_1,\ldots,(-1)^{f_n(b_1,\ldots,b_m)}P_n\rangle.\]
        Let $\ket{\psi(b_1,b_2,\ldots,b_m)}$ be a stabilizer state of $\cS_{b_1,b_2,\ldots,b_m}$, we have
        \[\bigl((-1)^{b\oplus f_{j_1}(b_1,\ldots,b_m)\oplus\cdots\oplus f_{j_k}(b_1,\ldots,b_m)}Z_q\bigr) \ket{\psi(b_1,b_2,\ldots,b_m)} = \ket{\psi(b_1,b_2,\ldots,b_m)}.\]
        Then,
        \[ Z_q \ket{\psi(b_1,b_2,\ldots,b_m)} = \bigl((-1)^{b\oplus f_{j_1}(b_1,\ldots,b_m)\oplus\cdots\oplus f_{j_k}(b_1,\ldots,b_m)}\bigr) \ket{\psi(b_1,b_2,\ldots,b_m)},\]
        which means $\ket{\psi(b_1,b_2,\ldots,b_m)}$ is an eigenvector of $Z_q$.
        The computational basis measurement on qubit $q$ doesn't change the state and implies a determinate measurement outcome $b\oplus f_{j_1}(b_1,\ldots,b_m)\oplus\cdots\oplus f_{j_k}(b_1,\ldots,b_m)$.
        Therefore, in this case, the measurement function $m$ for $\srho$ with computational basis measurement on qubit $q$ is defined as
        \[m(q, \srho) = \bigl(s, \delta_{s, b\oplus f_{j_1}(s_1,\ldots,s_m)\oplus\cdots\oplus f_{j_k}(s_1,\ldots,s_m)}, \srho\bigr)\]
        with $s$ being a newly introduced Boolean symbol for measurement outcome and $\delta_{s, e}$ being $1$ or $0$ if $s = e$ or $s\neq e$, respectively.
        Since the measurement outcome is determinate, we can also omit the introduced symbols by defining
        \[m(q, \srho) = (b\oplus f_{j_1}(s_1,\ldots,s_m)\oplus\cdots\oplus f_{j_k}(s_1,\ldots,s_m), \{\}, \srho).\]
        \item $Z_q$ anti-commutes with one or more of $P_j$.
        In this case, without loss of generality, we can assume $P_1$ anti-commutes with $Z_q$ and $P_2,\ldots,P_n$ commute with $Z_q$\footnote{If there is other $P_j, j\geq 2$ anti-commutes with $Z_q$, we can replace the $(-1)^{f_j(s_1,\ldots,s_m)}P_j$ by $(-1)^{f_1(s_1,\ldots,s_m)+f_j(s_1,\ldots,s_m)}P_1P_j$ in the generating set of $\srho$, and it will result in the same $\srho$. Then $Z_q$ only commute with $P_1P_j$.}.
        For any concrete values $b_1,\ldots,b_m$ of symbols $s_1,\ldots,s_m$, let $\ket{\psi(b_1,\ldots,b_m)}$ be a stabilizer state of 
        \[\langle (-1)^{f_1(b_1,\ldots,b_m)}P_1,\ldots,(-1)^{f_n(b_1,\ldots,b_m)}P_n\rangle.\]
        According to~\cite[Chapter 10.5.3]{nielsen2010quantum},
        the computational basis measurement on qubit $q$ will produce measurement outcomes $b = 0, 1$ with the same probability; and the measured state will become the stabilizer state of stabilizer
        \[\langle (-1)^{b}Z_q, (-1)^{f_2(b_1,\ldots,b_m)}P_2,\ldots,(-1)^{f_n(b_1,\ldots,b_m)}P_n\rangle.\]
        Therefore, in this case, the measurement function $m$ for $\srho$ with computational basis measurement on qubit $q$ is defined as
        \[m(q, \srho) = \left(s, \frac{1}{2}, \langle (-1)^{s}Z_q, (-1)^{f_2(s_1,\ldots,s_m)}P_2,\ldots,(-1)^{f_n(s_1,\ldots,s_m)}P_n\rangle\right),\]
        where $s$ is a newly introduced Boolean symbol for measurement outcome.
    \end{enumerate}
\end{itemize}

\subsection{Proof of \cref{lem:adequacy}}~

\adequacylem*
\begin{proof}
    Consider a fixed sequence of transitions
    \begin{equation}\label{eq:fixed}
        \cfg{S}{\sigma}{\rho} \overto{p_1} \cfg{S_1}{\sigma_1}{\rho_1} \overto{p_2} \cdots \overto{p_{n-1}} \cfg{S_{n-1}}{\sigma_{n-1}}{\rho_{n-1}} \overto{p_n} \cfg{S_{n}}{\sigma_{n}}{\rho_{n}} = \cfg{\terminate}{\sigma'}{\rho'}
    \end{equation}
    such that $n\geq 0$ and $p = \prod_{j=1}^np_j$.
    According to operational semantics in \cref{fig:operational}, we can see that for any $1\leq j\leq n$, there is a matrix $E_j$, the form of which is one of $\Id, U_{\bar{q}}, \ketbra[q]{0}{0}$ and $\ketbra[q]{1}{1}$, such that
    \[p_j\rho_j = E_j\rho_{j-1}E_{j}^{\dagger}.\]
    Let $E = E_nE_{n-1}\cdots E_{2}E_1$, we have
    \[ p\rho' = p\rho_n = E\rho E^{\dagger}\]
    For any $\rho = \ketbra[L]{x_1,x_2,\ldots,x_k}{x_1,x_2,\ldots,x_k}$, we divide the discussion into two cases:
    \begin{itemize}
        \item If $p> 0$, by \cref{transition-stab}, we have $\rho' = \rho = \ketbra[L]{x_1,x_2,\ldots,x_k}{x_1,x_2,\ldots,x_k}$. Then, 
        \[ p\ketbra[L]{x_1,x_2,\ldots,x_k}{x_1,x_2,\ldots,x_k} =  E\ketbra[L]{x_1,x_2,\ldots,x_k}{x_1,x_2,\ldots,x_k} E^{\dagger}\]
        which implies
        \[E\ket{x_1,x_2,\ldots,x_k}_L = c_{E,x_1,x_2,\ldots,x_k}\ket{x_1,x_2,\ldots,x_k}_L\]
        with $c_{E,x_1,x_2,\ldots,x_k}$ being a complex number.
        \item If $p = 0$, we have $E\ketbra[L]{x_1,x_2,\ldots,x_k}{x_1,x_2,\ldots,x_k} E^{\dagger} = 0$, then
        \[E\ket{x_1,x_2,\ldots,x_k}_L = c_{E,x_1,x_2,\ldots,x_k} \ket{x_1,x_2,\ldots,x_k}_L\]
        with $c_{E,x_1,x_2,\ldots,x_k} = 0$.
    \end{itemize}
    Therefore, we have that for any $\ket{x_1,x_2,\ldots,x_k}_L$, there is a complex number $c_{E,x_1,x_2,\ldots,x_k}$ such that
    \[E\ket{x_1,x_2,\ldots,x_k}_L = c_{E,x_1,x_2,\ldots,x_k} \ket{x_1,x_2,\ldots,x_k}_L.\]
    Similarly, for any $\bar{H}\ket{x_1,x_2,\ldots,x_k}_L$, there is a complex number $d_{E,x_1,x_2,\ldots,x_k}$ such that
    \[E\bar{H}\ket{x_1,x_2,\ldots,x_k}_L = d_{E,x_1,x_2,\ldots,x_k} \bar{H}\ket{x_1,x_2,\ldots,x_k}_L.\]
    Then, for any $x_1,x_2,\ldots,x_k\in\{0,1\}^k$, consider
    {\allowdisplaybreaks \begin{align*}
        &\sum_{y_1,y_2,\ldots,y_k\in\{0,1\}^k}\frac{1}{\sqrt{2^k}}(-1)^{x_1y_1+x_2y_2+\cdots+x_ky_k}d_{E,x_1,x_2,\ldots,x_k}\ket{y_1,y_2,\ldots,y_k}_{L} \\\
        ={}& d_{E,x_1,x_2,\ldots,x_k}\sum_{y_1,y_2,\ldots,y_k\in\{0,1\}^k}\frac{1}{\sqrt{2^k}}(-1)^{x_1y_1+x_2y_2+\cdots+x_ky_k}\ket{y_1,y_2,\ldots,y_k}_{L} \\
        ={}& d_{E,x_1,x_2,\ldots,x_k} \bar{H}\ket{x_1,x_2,\ldots,x_k}_L \\
        ={}& E\bar{H}\ket{x_1,x_2,\ldots,x_k}_L \\
        ={}& E\sum_{y_1,y_2,\ldots,y_k\in\{0,1\}^k}\frac{1}{\sqrt{2^k}}(-1)^{x_1y_1+x_2y_2+\cdots+x_ky_k}\ket{y_1,y_2,\ldots,y_k}_{L} \\
        ={}& \sum_{y_1,y_2,\ldots,y_k\in\{0,1\}^k}\frac{1}{\sqrt{2^k}}(-1)^{x_1y_1+x_2y_2+\cdots+x_ky_k}E\ket{y_1,y_2,\ldots,y_k}_{L} \\
        ={}& \sum_{y_1,y_2,\ldots,y_k\in\{0,1\}^k}\frac{1}{\sqrt{2^k}}(-1)^{x_1y_1+x_2y_2+\cdots+x_ky_k}c_{E,y_1,y_2,\ldots,y_k}\ket{y_1,y_2,\ldots,y_k}_{L}.
    \end{align*}}
    Since $\{\ket{y_1,y_2,\ldots,y_k}_{L}|y_1,y_2,\ldots,y_k\in\{0,1\}^k\}$ forms a complete orthonormal basis of the code space, we have that any $y_1,y_2,\ldots,y_k \in \{0,1\}^k$,
    \[c_{E,y_1,y_2,\ldots,y_k} = d_{E,x_1,x_2,\ldots,x_k}.\]
    Thus, there is a complex number $c$ such that for any $x_1,x_2,\ldots,x_k \in \{0,1\}^k$,
    \[c_{E,x_1,x_2,\ldots,x_k} = d_{E,x_1,x_2,\ldots,x_k} = c.\]
    Consider any pure state $\ket{\psi}$ in the code space, it admits a superposition of logical computational basis: 
    \[\ket{\psi} = \sum_{x_1,x_2,\ldots,x_k\in\{0,1\}^k}\alpha_{x_1,x_2,\ldots,x_k}\ket{x_1,x_2,\ldots,x_k}_L.\]
    Then,
    \begin{align*}
        E\ket{\psi} &= E \sum_{x_1,x_2,\ldots,x_k\in\{0,1\}^k}\alpha_{x_1,x_2,\ldots,x_k}\ket{x_1,x_2,\ldots,x_k}_L \\
        &=\sum_{x_1,x_2,\ldots,x_k\in\{0,1\}^k}\alpha_{x_1,x_2,\ldots,x_k}E\ket{x_1,x_2,\ldots,x_k}_L \\
        &= \sum_{x_1,x_2,\ldots,x_k\in\{0,1\}^k}\alpha_{x_1,x_2,\ldots,x_k}c\ket{x_1,x_2,\ldots,x_k}_L \\
        &=c \sum_{x_1,x_2,\ldots,x_k\in\{0,1\}^k}\alpha_{x_1,x_2,\ldots,x_k}\ket{x_1,x_2,\ldots,x_k}_L \\
        &= c\ket{\psi} 
    \end{align*}
    Now, consider any quantum state $\rho$ in the code space, it admits a spectral decomposition
    \[\rho = \sum_{\lambda}\lambda\ketbra{\psi_{\lambda}}{\psi_{\lambda}}\]
    with $\ket{\psi_{\lambda}}$ being a pure state in code space, and the fixed sequence of transitions of \cref{eq:fixed}, we have
    \[p\rho' = E\rho E_{\dagger} = E \left(\sum_{\lambda}\lambda\ketbra{\psi_{\lambda}}{\psi_{\lambda}}\right) E^{\dagger} = \sum_{\lambda}\lambda E\ketbra{\psi_{\lambda}}{\psi_{\lambda}}E^{\dagger} = \sum_{\lambda}\lambda cc^{*}\ketbra{\psi_{\lambda}}{\psi_{\lambda}} = cc^{*} \rho\]
    If $p > 0$, then $p = p\tr(\rho') = \tr(p\rho') = \tr(cc^{*}\rho) = cc^{*}\tr(\rho) = cc^{*}$. Therefore,
    \[\rho' = \frac{cc^{*}}{p}\rho = \rho.\]
\end{proof}

\subsection{Proof of \cref{thm:adequacy}}~

\adequacythm*
\begin{proof}
    Based on the above context in this section, we choose:
    \begin{enumerate} \item $\srho_1$ be the symbolic stabilizer state $\langle P_1,\ldots,P_{n-k},(-1)^{s_1}L_1,\ldots,(-1)^{s_k}L_k\rangle$; and \item $\srho_2$ be the symbolic stabilizer state $\langle \bar{H}P_1\bar{H},\ldots,\bar{H}P_{n-k}\bar{H},(-1)^{s_1}\bar{H}L_1\bar{H},\ldots,(-1)^{s_k}\bar{H}L_k\bar{H}\rangle$ with $s_j$ symbols over Boolean values.
    \end{enumerate}
    For any $\rho \in \mathcal{T}$ in \cref{lem:adequacy},
    without loss of generality, we can assume that \[\rho = \ketbra[L]{x_1,s_2,\ldots,x_k}{x_1,s_2,\ldots,x_k}.\]
    Let $V_1 = V[s_1\mapsto x_1,s_2\mapsto x_2,\ldots, s_k\mapsto x_k]$, we have
    $V_1(\srho_1) = \rho$ and $V_1(\ssigma) = V(\ssigma)$ since $\ssigma$ does not contains $s_j, 1\leq j\leq k$. Therefore, by \cref{def:satisfaction}, we have
    \begin{equation}\label{eq:s1}
        (V(\ssigma), \rho) \models (\ssigma, \srho_1).
    \end{equation}

    For any sequence of (operational semantics) transitions
    \begin{equation*}\label{eq:fixed2}
        \cfg{S}{V(\ssigma)}{\rho} \overto{p_1} \cfg{S_1}{\sigma_1}{\rho_1} \overto{p_2} \cdots \overto{p_{n-1}} \cfg{S_{n-1}}{\sigma_{n-1}}{\rho_{n-1}} \overto{p_n} \cfg{\terminate}{\sigma'}{\rho'}
    \end{equation*}
    such that $n\geq 0$ and $p = \prod_{j=1}^np_j > 0$, we have a corresponding sequence of (symbolic execution) transitions:
    \[\scfg{S}{\ssigma}{\srho_1}{\emptyset}{\<true>} \to \scfg{S_1}{\ssigma_1}{\srho_{1}^{1}}{P_1}{\varphi_1} \to \cdots \to \scfg{S_{n-1}}{\ssigma_{n-1}}{\srho_{1}^{n-1}}{P_{n-1}}{\varphi_{n-1}} \to \scfg{\terminate}{\ssigma'}{\srho'}{P'}{\varphi'}\]
    With the proof of \cref{thm:soundess}, \cref{eq:s1} and induction through $n$, we can obtain a valuation $V_2$ such that 
    \[
        \begin{array}{lll}
            V_2(\ssigma) = V(\sigma), & V_2(\srho_1) = \rho, \\
            V_2(\ssigma_1) = \sigma_1, & V_2(\srho_1^1) = \rho_1, & V_2(\varphi_1) = \<true>, \\
            \quad \vdots &\cdots \\
            V_2(\ssigma') = \sigma', & V_2(\srho') = \rho', & V_2(\varphi') = \<true>
        \end{array}
    \]
    Moreover, the sequence of (symbolic execution)
    transitions imply
    \[\scfg{S}{\ssigma}{\srho_1}{\emptyset}{\<true>} \to^{*} \scfg{\terminate}{\ssigma'}{\srho'}{P'}{\varphi'},\]
    then,
    \begin{equation}
        \label{eq:s3}
        \varphi' \models\srho' = \srho_1
    \end{equation}
    Since $V_2(\varphi') = \<true>$, we have $V_2(\srho') = V_2(\srho_1)$, then $\rho' = V_2(\srho') = V_2(\srho_1) = \rho$.
    Therefore, we have proved that for any quantum state $\rho\in \mathcal{T}$, 
    \[\cfg{S}{V(\ssigma)}{\rho} \overto{p}^{*} \cfg{\terminate}{\sigma'}{\rho'} \text{ with $p>0$ implies $\rho' = \rho$}.\]
    Then, by \cref{lem:adequacy}, it holds for any quantum state $\rho$ in the code space.
\end{proof}

\end{document}